\documentclass[a4paper,english,cleveref, thm-restate, authorcolumns]{lipics-v2021}

\usepackage{nicefrac}

\graphicspath{{./figures/}}

\title{Minimum Scan Cover and Variants - \texorpdfstring{\newline}{} Theory and Experiments}

\nolinenumbers

\titlerunning{Minimum Scan Cover and Variants}

\author{Kevin Buchin}{Department of Mathematics \& Computer Science\\ TU Eindhoven, Eindhoven, The Netherlands}{k.a.buchin@tue.nl}{https://orcid.org/0000-0002-3022-7877}{}
\author{S\'{a}ndor P. Fekete}{Department of Computer Science\\ TU Braunschweig, Braunschweig, Germany}{s.fekete@tu-bs.de}{https://orcid.org/0000-0002-9062-4241}{}
\author{Alexander Hill}{Department of Computer Science\\ TU Braunschweig, Braunschweig, Germany}{a.hill@tu-bs.de}{https://orcid.org/0000-0002-9270-9871}{}
\author{Linda Kleist}{Department of Computer Science\\ TU Braunschweig, Braunschweig, Germany}{l.kleist@tu-bs.de}{https://orcid.org/0000-0002-3786-916X}{}
\author{Irina Kostitsyna}{Department of Mathematics \& Computer Science\\ TU Eindhoven, The Netherlands}{i.kostitsyna@tue.nl}{https://orcid.org/0000-0003-0544-2257}{}
\author{Dominik Krupke}{Department of Computer Science\\ TU Braunschweig, Braunschweig, Germany}{d.krupke@tu-bs.de}{https://orcid.org/0000-0003-1573-3496}{}
\author{Roel Lambers}{Department of Mathematics \& Computer Science\\ TU Eindhoven, The Netherlands}{r.lambers@tue.nl}{https://orcid.org/0000-0002-0314-6094}{}
\author{Martijn Struijs}{Department of Mathematics \& Computer Science\\ TU Eindhoven, The Netherlands}{m.a.c.struijs@tue.nl}{https://orcid.org/0000-0002-0116-7238}{}

\authorrunning{Buchin, Fekete, Hill, Kleist, Kostitsyna, Krupke, Lambers and Struijs} 

\Copyright{Kevin Buchin, S\'andor P.\,Fekete, Alexander Hill, Linda Kleist, Irina Kostitsyna, Dominik Krupke, and Martijn Struijs} 

\ccsdesc[100]{{Theory of computation $\rightarrow$ Design and analysis of algorithms, Computational geometry; Applied computing → Operations research}} 

\keywords{Graph scanning, angular metric, makespan, energy, bottleneck, complexity, approximation, algorithm engineering, mixed-integer programming, constraint programming} 

\category{} 

\supplement{\url{https://github.com/ahillbs/minimum_scan_cover}}

\funding{Work at TU Braunschweig was partially supported under grant FE407/21-1, ``Computational Geometry: Solving Hard Optimization Problems'' (CG:SHOP).}

\hideLIPIcs  

\EventLongTitle{arXiv}
\EventShortTitle{arXiv}
\EventAcronym{arXiv}

\usepackage{xspace}
\usepackage[binary-units=true]{siunitx} 

\newcommand{\MSClength}{MSC-MS\xspace}
\newcommand{\MSCtotEnergy}{MSC-TE\xspace}
\newcommand{\MSClocEnergy}{MSC-BE\xspace}

\newcommand{\MIPone}{MIP-1\xspace}
\newcommand{\MIPtwo}{MIP-2\xspace}
\newcommand{\MIPthree}{MIP-3\xspace}
\newcommand{\CPone}{CP-1\xspace}
\newcommand{\CPtwo}{CP-2\xspace}

\newcommand{\cone}{cone\xspace}
\newcommand{\cupdot}{\mathbin{\mathaccent\cdot\cup}}
\newcommand{\mincover}{$\Lambda$-cover\xspace}

\crefname{claim}{Claim}{Claims}
\crefname{figure}{Figure}{Figures}

\begin{document}

\maketitle

\begin{abstract}
  We consider a spectrum of geometric optimization problems motivated by
contexts such as satellite communication and astrophysics.
In the problem \textsc{Minimum Scan Cover with Angular Costs}, we are given a
graph $G$ that is embedded in Euclidean space. The edges of $G$ need to be scanned,
i.e., probed from both of their vertices. In order to scan their edge, two
vertices need to face each other; changing the heading of a vertex incurs
some cost in terms of energy or rotation time that is proportional to the
corresponding rotation angle. Our goal is to compute schedules that minimize the following objective
functions: 
(i) in \textsc{Minimum Makespan Scan Cover} (\MSClength), this is the time until all edges are scanned; (ii) in \textsc{Minimum Total Energy Scan Cover} (\MSCtotEnergy),
the sum of all rotation angles; (iii) in \textsc{Minimum Bottleneck Energy Scan Cover} (\MSClocEnergy),
the maximum total rotation angle at one vertex.

Previous theoretical work on \MSClength revealed a close connection to graph coloring and the cut cover problem,
leading to hardness and approximability results. 
In this paper, we present polynomial-time algorithms for 1D instances of \MSCtotEnergy and \MSClocEnergy,
but NP-hardness proofs for bipartite 2D instances. For bipartite graphs in 2D, we also give
2-approximation algorithms for both \MSCtotEnergy and \MSClocEnergy. 
Most importantly, we provide
a comprehensive study of practical methods for all three problems.
We compare three different mixed-integer programming and two constraint programming
approaches, and show how to compute provably optimal solutions for geometric instances
with up to \num{300} edges.
Additionally, we compare the performance of different meta-heuristics for even larger instances.

\end{abstract}
\newpage
\section{Introduction}
For many aspects of wireless communication, the relative direction,
i.e., the angle of visibility between different locations, plays a crucial role.
A particularly
striking example occurs in the context of inter-satellite communication,
which requires focused transmission, with communication partners facing each
other with directional, paraboloid antennas or laser beams. This makes it
impossible to exchange information with multiple partners at once. Moreover, a
change of communication partner requires a change of heading, which is costly
in the context of space missions with limited resources, making it worthwhile
to invest in good schedules. Problems of this type do not only
arise from long-distance communication. They also come into play when astro-
and geophysical measurements are to be performed, in which groups of spacecraft
can determine physical quantities not just at their current locations, but also
along their common line of sight; see~\cite{korth2002particle} for a description.

In previous theoretical work~\cite{DBLP:conf/compgeom/FeketeKK20},
we considered an optimization problem arising from this context: How can we
schedule a given set of intersatellite communications, such that the overall
timetable is as efficient as possible? In the problem
\textsc{Minimum Scan Cover with Angular Costs} (MSC), the task is to establish a
collection of connections between a given set of locations, described by a
graph $G = (V,E)$ that is embedded in space. For any connection (or scan) of an
edge, the two involved vertices need to face each other; changing the heading
of a vertex to cover a different connection takes an amount of time
proportional to the corresponding rotation angle. In~\cite{DBLP:conf/compgeom/FeketeKK20},
the goal considered was to minimize the time
until all tasks are completed, i.e., compute a geometric schedule of minimum
makespan.

Given the importance of conserving energy on
space (or drone) missions, this \textsc{Minimum Makespan Scan Cover} (\MSClength)
is not the only important objective:
In \textsc{Minimum Total Energy Scan Cover} (\MSCtotEnergy), the goal is to minimize
the sum of all rotation angles; in \textsc{Minimum Bottleneck Energy Scan Cover} (\MSClocEnergy),
the task is to limit the energy used by any one vertex by
minimizing the maximum total rotation at one vertex.

In this paper, we complement the previous theoretical results on \MSClength
(hardness and approximation)~\cite{DBLP:conf/compgeom/FeketeKK20} by presenting an NP-hardness proof and a 2-approximation for \MSCtotEnergy and \MSClocEnergy for bipartite graphs in two dimensions.
For one-dimensional instances of \MSCtotEnergy and \MSClocEnergy, we show a polynomial time algorithm and an upper bound independent of the chromatic number, which shows a fundamental difference to \MSClength.
Most importantly, we provide
a comprehensive study of practical methods for all three objective functions.
We compare three different mixed-integer programming (MIP) and two constraint programming (CP)
approaches, and show how to compute provably optimal solutions for geometric instances
with up to \num{300} edges.
Additionally, we evaluate the practical performance of approximation algorithms and heuristics for even larger instances.

\subsection{Previous Work}
The use of directional antennas has introduced a number of geometric questions.
The paper at hand expands on previous work of
Fekete, Kleist, and Krupke~\cite{DBLP:conf/compgeom/FeketeKK20}, who investigated \MSClength and identified a close connection to graph coloring and the (directed) cut cover number.
More precisely,  \MSClength in 1D and 2D is in $\Theta(\log \chi(G))$, which
implies that even in 1D, there exists no constant-factor approximation for \MSClength. 
For 2D, they present a 4.5-approximation for bipartite instances
and show inapproximability for a constant better than $\nicefrac{3}{2}$. This
yields an $O(c)$-approximation for $k$-colored graphs with $k\leq \chi(G)^c$.

Further problems involving directional antennas have been considered by
Carmi et al.~\cite{carmi2011connectivity}, who study
the $\alpha$-MST problem. This problem arises from finding orientations of directional antennas with $\alpha$-cones, such
that the connectivity graph yields a spanning tree of minimum weight, based on bidirectional communication.
They prove that for $\alpha<\nicefrac \pi 3$, a solution may not exist, while $\alpha\geq\nicefrac \pi 3$ always suffices.
See Aschner and Katz~\cite{aschner2017bounded} for more recent hardness proofs and constant-factor
approximations for some specific values of $\alpha$.

Many other geometric optimization problems deal with turn
cost. Arkin et al.~\cite{arkin2001optimal,arkin2005optimal}
show hardness
of finding an optimal milling tour with turn cost, even in relatively constrained settings,
and give approximation algorithms.
The complexity of finding an optimal cycle cover in a 2-dimensional grid
graph was stated as \emph{Problem~{53}} in \emph{The Open Problems
	Project}~\cite{openproblemproject} and shown
to be NP-hard in~\cite{fk-ctcct-19}, which also provides constant-factor
approximations; practical methods and results are given in~\cite{ALENEX19},
and visualized in the video~\cite{bdf+-zzmmd-17}.

Finding a fastest roundtrip for a set of points in the plane for which the
travel time depends only on the turn cost is called the \textsc{Angular Metric
	Traveling Salesman Problem}. Aggarwal et al.~\cite{aggarwal2000angular} prove
hardness and provide an $O(\log n)$ approximation algorithm.
For the abstract version on graphs in which ``turns'' correspond to weighted changes between edges,
Fellows et al.~\cite{fellows2009abstract} show that the problem is fixed-parameter tractable in the number of turns, the treewidth, and the maximum degree.
Fekete and Woeginger~\cite{fekete1997angle} consider the problem of connecting a set of points
by a tour in which the angles of successive edges are
constrained.

\MSClength is a special case of scheduling in which the cost of a current job
depends on the sequence of the already processed ones; e.g., Allahverdi et
al.~\cite{journals/eor/Allahverdi15,allahverdi1999review,allahverdi2008survey}
provide a comprehensive overview, especially on practical work.  In the context
of earth observation, Li et al.~\cite{li2017hybrid} and Augenstein et
al.~\cite{augenstein2016optimal} describe MIPs and heuristics to schedule image
acquisition and downlink for satellites for which rotation and setup costs are
taken into account.

\subsection{Preliminaries and Problem Definitions}
For all considered versions of \textsc{Minimum Scan Cover} (\textsc{MSC}),
the \emph{input} consists of a (straight-line) embedded (not necessarily crossing-free) graph $G=(V,E)$ with a finite vertex set $V
\subset \mathbb{R}^2$. We refer to the elements of $V$ as \emph{points}
when their specific locations in $\mathbb{R}^2$ are relevant; if we focus
on graph properties, we may also refer to them as \emph{vertices}. We denote
the undirected edge between $u, v \in V$ by $uv$.
For $v \in V$, we let $N(v) = \{u \in V : uv \in E\}$ be all vertices adjacent to $v$, and $E(v) = \{ uv : u\in N(v)\}$ be all edges incident to $v$. For two adjacent edges $uv, vw \in E(v)$, let $\alpha( uv, vw ) \in [0,\ang{180}]$ denote the smaller angle
between the lines supporting the segments $uv$ and $vw$.
The \emph{output} for each problem is a \emph{scan cover} $S : E \rightarrow \mathbb{R}^+$, such that
 for all pairs of adjacent edges  $e, e'$, we have $|S(e) - S(e')| \geq \alpha(e,e')$.
The geometric interpretation of a scan cover is that all points $v \in V$ have a \textit{heading} that can change over time, and that if $S(uv) = t$ then $u$ and $v$ \emph{face} each other at time~$t$. In this case, we say that the edge $uv$ is \emph{scanned} at time~$t$. Thus, the above condition on $S$ guarantees that $S$ complies with the necessary rotation time if rotation speed is bounded by~$1$.

A \emph{rotation scheme} describes the geometric change of headings of the
vertices over a time interval of length $T$, i.e., it is a map $r\colon V\times
[0,T]\mapsto [ \ang{0}, \ang{360}]$. 
The \emph{total rotation angle} of a vertex $v$ in $r$ is the total amount that $v$ rotates over $[0,T]$.
For a given scan cover $S$, we are particularly interested in edges that are scanned consecutively. Therefore, let $\nu_S(e,e') = 1$ if $e$ and $e'$ share exactly one vertex $v$ and the edge $e'$ is scanned directly after $e$ at $v$; otherwise $\nu_S(e,e') = 0$.

\subparagraph{The Problems}
We consider the following three problems, defined by their respective objectives.
For a given graph $G=(V,E)$ with vertices in the plane, find a scan cover $S$ with
\begin{itemize}
	\item Minimum Makespan (\MSClength): \qquad \qquad $\displaystyle \min \max_{e \in E} S(e)$
	\item {Minimum Total Energy (\MSCtotEnergy):} \qquad \ \ $\displaystyle \min \sum_{v\in V} \sum_{e,e' \in E(v)}  \alpha(e,e')\cdot \nu_S(e,e')$
	\item {Minimum Bottleneck Energy (\MSClocEnergy):} \ $\displaystyle \min \max_{v \in V} \sum_{e,e' \in E(v)} \alpha(e,e') \cdot \nu_S(e,e')$
\end{itemize}

	Concentrating on the expensive and algorithmically challenging part of efficient rotations between the edges, we do not fix the initial heading of the satellites. In fact, all algorithms can be easily adapted to handle fixed initial headings.
  Furthermore, for every of the three objectives,  an $f$-approximation  can be converted into a $(f+1)$-approximation for the problem variant with fixed initial headings.

For a vertex $v$, we  denote by $\Lambda(v)$ the minimum angle, such
that a cone of this angle with apex $v$ contains all edges in $E(v)$.
We call such a cone a $\Lambda$-\cone of $v$ and call the complement of such a cone an \emph{outer} \cone of $v$.
A \emph{\mincover} is a scan cover for which every vertex $v$ rotates in a single direction, either clockwise or counterclockwise, with a total rotation angle equal to $\Lambda(v)$. Note that different vertices can rotate in different directions. A \mincover minimizes both the \MSCtotEnergy and \MSClocEnergy objectives.

\subsection{Outline and Results}

This paper consists of a theoretical part (\cref{sec:complexity}) and a practical part (\cref{sec:practical}).
Our theoretical results complement the work on \MSClength~\cite{DBLP:conf/compgeom/FeketeKK20}
by hardness and approximation results for the two new objectives \MSCtotEnergy and \MSClocEnergy.
In \cref{sec:complexity}, we show that both problems can be solved efficiently in 1D;
on the other hand, we prove that they  are NP-hard in 2D, even for bipartite
graphs. Finally, we complement the hardness results
by providing  2-approximations for bipartite graphs and $O(\log
n)$-approximations for general graphs. Our
practical study in \cref{sec:practical} considers optimal solutions in
\cref{sec:practical:exact} and heuristic solutions in
\cref{sec:practical:heuristic}.
For optimal solutions, we develop three mixed integer linear programs (MIPs), as well as
two constraint programs (CPs) and evaluate their practical performance on
a suite of benchmark instances.  Solving instances of
\MSCtotEnergy and \MSClocEnergy to provable optimality turned out to be quite difficult; for \MSClength, we were able
to solve instances with up to \num{300} edges, based on one CP.
In addition, we compared the solution quality of four (meta-)heuristics and the approximation algorithms on larger instances with up to \num{800} edges.
In our experiments, a genetic algorithm and the intermediate solution after timeout of one CP produced the best solutions.

\section{Complexity results}
\label{sec:complexity}
Fekete, Kleist, and Krupke studied the computational complexity of \MSClength~\cite{DBLP:conf/compgeom/FeketeKK20}.
In this section, we provide new results for \MSClocEnergy and \MSCtotEnergy.

For \MSClength in 1D, when all vertices are placed on a line,  there exists no constant-factor approximation unless $P=NP$~\cite{DBLP:conf/compgeom/FeketeKK20}. In contrast, we show that \MSCtotEnergy and \MSClocEnergy in 1D can be solved efficiently.
Refer to Appendix~\ref{appendix:1d} for a proof of the following theorem.

\begin{restatable}{theorem}{onedsimple}\label{thm:1dsimple}
\MSCtotEnergy and \MSClocEnergy in 1D are in $P$.
Moreover, denoting by $k$ the number of vertices with neighbors to both sides, the objective value is $0$
for $k=0$, while for $k>0$ it is $\ang{180}\cdot k$ for \MSCtotEnergy and $\ang{180}$ for \MSClocEnergy.
\end{restatable}

Next we show that for 2D instances of \MSCtotEnergy and \MSClocEnergy, there does not exist an efficient algorithm, unless $P=NP$.
Specifically, we show that \MSCtotEnergy and \MSClocEnergy are NP-hard in 2D, even when the underlying graph $G=(V,E)$ is bipartite. Our proof is based on the observation that if a \mincover exists, any scan cover optimal for \MSCtotEnergy is a \mincover. If additionally all vertices have the same $\Lambda(v)$, any scan cover optimal for \MSClocEnergy is a \mincover. We show finding a \mincover is NP-hard via a reduction from the NP-complete problem \textsc{Monotone Not-all-equal 3-satisfiability} (\textsc{MNAE3SAT})~\cite{lovasz1973coverings,Schaefer78SAT}, defined as follows:
Given a set of Boolean variables $X$ and a set of clauses $\mathcal{C}$ with at most 3 literals from $X$ all of which are not negated, is there a 0/1-assignment to the variables in $X$, such that for each clause in $\mathcal{C}$, not all variables have the same value?

Given an instance $I$ of the \textsc{MNAE3SAT}, we construct an \textsc{MSC} instance $G_I$ (with the same $\Lambda(v)$ for all vertices) that has a \mincover if and only if $I$ has a valid variable assignment.
Recall that in a \mincover, the edges of each vertex are scanned in either clockwise or counter-clockwise order. We encode variable assignment by the rotation direction of the vertices in~$G_I$ in a \mincover.
A variable is encoded by a subgraph that contains a set of vertices that have the same rotation direction in a \mincover, and a clause by a subgraph that contains three vertices that cannot all have the same rotation direction in a \mincover. We connect variables to clauses via wires, which are encoded by a subgraph that contains two vertices that have the same rotation direction in a \mincover.
See Figure~\ref{fig:hardness-construction-example} for an example of the construction.
Refer to Appendix~\ref{appendix:nphardness} for the complete construction and a proof of the following theorem.
\begin{restatable}{theorem}{nphardness}\label{thm:NP-hard-2D}
	\MSCtotEnergy and \MSClocEnergy in 2D  are NP-hard, even for bipartite graphs.
\end{restatable}%
\begin{figure}[htbp]
	\centering
	\includegraphics[page=4]{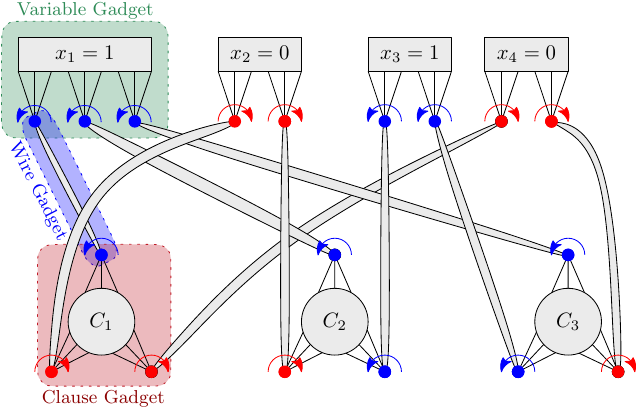}
	\caption{The constructed graph $G_I$ for
  the  instance $(x_1\vee x_2 \vee x_4)\wedge (x_1\vee x_2\vee x_3) \wedge
(x_1\vee x_3\vee x_4)$ of \textsc{MNAE3SAT}. The gadgets are drawn symbolically;
also shown are the directions of the connector vertices corresponding to
the satisfying assignment $x_1=x_3=1,x_2=x_4=0$.}
\label{fig:hardness-construction-example}
\end{figure}
The construction in the proof of Theorem~\ref{thm:NP-hard-2D} establishes a gap
between optimal and suboptimal solutions, which implies a constant-factor
approximation lower bound for \MSClocEnergy.
Refer to Appendix~\ref{appendix:nphardness} for a proof of the following corollary.

\begin{restatable}{corollary}{nphardcor}\label{cor:approximtion-hardness-2D}
\MSClocEnergy in 2D is NP-hard to approximate within a factor of $1.04$, even for bipartite graphs.
\end{restatable}

Next, we complement the  4.5-approximation algorithm for \MSClength in bipartite graphs in the plane~\cite{DBLP:conf/compgeom/FeketeKK20}
by presenting an approximation algorithms for both remaining objectives.

\begin{theorem}\label{thm:approxiBip}
There exists a 2-approximation algorithm for \MSClocEnergy and \MSCtotEnergy for each bipartite graph $G=(V_1\cupdot  V_2, E)$ embedded in the plane.
\end{theorem}

\begin{proof}
Defining $V:=V_1\cup V_2$, the values $\max_{v\in V}\Lambda(v)$ and $\sum_{v\in V}\Lambda(v)$ are clearly lower bounds on the value of a scan cover minimizing \MSClocEnergy and \MSCtotEnergy, respectively.
	
  We use the following geometric property based on alternating angles that is also used
in~\cite{DBLP:conf/compgeom/FeketeKK20}): Starting with opposite headings, two
vertices face their edge at the same time when both start a full clockwise
rotation simultaneously. Defining start headings
$r(v,0):=\ang{0}$ for $v\in V_1$ and  $r(v,0):=\ang{180}$ for $v\in V_2$,  the
clockwise rotation scheme induces a scan cover~$S$ by defining the scan time
$S(e)$ of edge $e$ as  the time when its two vertices face each other.

We now show that in the rotation scheme $r'$ induced by $S$, i.e., every vertex $v$ starts to head towards its edge  first scanned in $S$ and then follows the order on $E(v)$ defined by $S$, the total rotation angle of each vertex $v$ is at most $2\Lambda(v)$. To this end, we consider three types of vertices; for an illustration consider \cref{fig:Approxi}.

	\begin{figure}[htb]
	\centering
	\begin{subfigure}[t]{.3\textwidth}
		\centering
		\includegraphics[page=4]{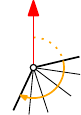}
		\caption{Case (a).
		}
		\label{fig:ApproxiA}
	\end{subfigure}\hfil
	\begin{subfigure}[t]{.3\textwidth}
		\centering
		\includegraphics[page=3]{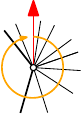}
    \caption{Case (b).
}
		\label{fig:ApproxiB}
	\end{subfigure}\hfil
	\begin{subfigure}[t]{.3\textwidth}
		\centering
		\includegraphics[page=2]{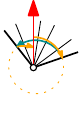}
    \caption{Case (c).
    }
		\label{fig:ApproxiC}
	\end{subfigure}
	\caption{Illustration for the rotation scheme $r'$ in the proof of \cref{thm:approxiBip}.}
	\label{fig:Approxi}
\end{figure}

	\begin{alphaenumerate}
	\item \label{item:a} Case: $r(v,0)$ lies outside the $\Lambda$-\cone of $v$. Then all edges of $v$ are scanned by a clockwise rotation, one after the other. Hence, $v$
	 has a total rotation angle of $\Lambda(v)$.
 \item Case: $r(v,0)$ lies inside the $\Lambda$-\cone of $v$ and $\Lambda(v)\geq \ang{180}$.
 Going over all edges clockwise takes at most a full rotation of $\ang{360}\leq 2\Lambda(v)$.
 \item Case: $r(v,0)$ lies inside the $\Lambda$-\cone of $v$ and $\Lambda(v)<\ang{180}$.  Let $e_1$ and $e_2$ denote the bounding edges of the $\Lambda$-\cone  such that $S(e_1)\leq S(e_2)$.
   By definition, the minimal angle of $e_1$ and $e_2$ is $\Lambda(v)<\ang{180}$. Splitting the $\Lambda$-\cone   of $v$ into two halves at $r(v,0)$, $v$ scans the edges in each half  in clockwise direction, rotating an angle of $\Lambda(v)$ counterclockwise between $e_1$ and $e_2$.
	It follows that the total rotation angle of $v$ is at most $2\Lambda(v)$.
\end{alphaenumerate}

As the total rotation angle is at most $2\Lambda(v)$ for each vertex $v$, \MSClocEnergy and \MSCtotEnergy are upper bounded by $2\cdot \max_{v\in V} \Lambda(v)$ and $2\cdot \sum_{v\in V}\Lambda(v)$. Together with the lower bounds provided above, this shows that this scan cover is a 2-approximation for \emph{either} objective. 
\end{proof}

\begin{corollary}\label{cor:optimal}
Let $G=(V_1\cupdot  V_2, E)$ be a bipartite graph embedded in the plane such that the points of $V_1$ and $V_2$ can be separated by a line. Then an optimal \MSClocEnergy and \MSCtotEnergy of $G$ can be found in polynomial time.
\end{corollary}
\begin{proof}
	We follow the same technique as in the proof of \cref{thm:approxiBip}.
	We may assume without loss of generality that the separating line is vertical and that the points of $V_1$ lie left of the line. Then, with the above definitions, every vertex is in case (a), i.e., the total rotation angle for each vertex $v$ is $\Lambda(v)$. Consequently, the resulting scan cover is optimal for both \MSClocEnergy and \MSCtotEnergy.
\end{proof}

The insights from \cref{thm:approxiBip} yield an approximation algorithm for $k$-colored graphs.
\begin{corollary}\label{cor:specialCases}
  For  \MSCtotEnergy and \MSClocEnergy of $k$-colored graphs embedded in 2D, there exists an $O\left(\log\left(k\right)\right)$-approximation. 
\end{corollary}
\begin{proof}
    The edges of a $k$-colored graph $G$ can be covered by $\lceil \log_2(k)\rceil$ bipartite graphs~$G_i$~\cite{Motwani:1994:EAC:891890}.
  For each $G_i$, we use the 2-approximation of \cref{thm:approxiBip}.
  Clearly, for both objectives, the optimal scan time of $G$ is lower  bounded by the optimal scan time of every subgraph. Consequently, scanning all $G_i$ takes at most $\sum_i 2\cdot OPT(G_i)\leq 2\lceil \log_2(k)\rceil\cdot OPT(G)$, where $OPT$ denotes the optimum scan time for the respective objective.
  For adjusting the headings between the scan covers of the bipartite graphs, we need  $(\lceil \log_2(k)\rceil - 1)$ transition phases each of which needs at most $OPT(G)$. Hence, the total scan time is upper bounded by $(3\lceil \log_2(k)\rceil - 1)\cdot OPT(G)$
\end{proof}

\section{Experiments}
\label{sec:practical}

For our experimental evaluation, we considered two types of benchmark instances in 2D, which we call \emph{random} and \emph{celestial}.
Random instances are generated by placing $n$ points chosen uniformly at random from the unit square, with each edge chosen with probability $p$.
Note that the visible area of a satellite constellation on the same altitude in low Earth orbit is fairly close to a set of co-planar points and hence the square (or plane) serves as a reasonable approximation.

Celestial instances are inspired by real-world instances of satellites in a shared orbit,
in which they maintain their relative positions while orbiting around a central body like Earth, as long as no
explicit orbit-changing maneuvers are carried out.
They are characterized by a set of points on a circle and a central circular obstacle.
The points on the circle are chosen uniformly at random; an edge exists if and only if its
vertices \emph{see} each other, i.e., the edge does not intersect the central obstacle.
Examples and the distribution of the nearly \num{2000} instances with up to \num{800} edges used for
our experiments can be seen in \cref{fig:preliminaries:instanceexamples}.

All experiments were run on Intel Core i7-3770 with \SI{3.4}{\GHz} and \SI{32}{\giga\byte} of RAM.
\begin{figure}[htb]
  \centering
  \begin{subfigure}[t]{.25\textwidth}
    \centering
    \includegraphics[width=0.9\textwidth]{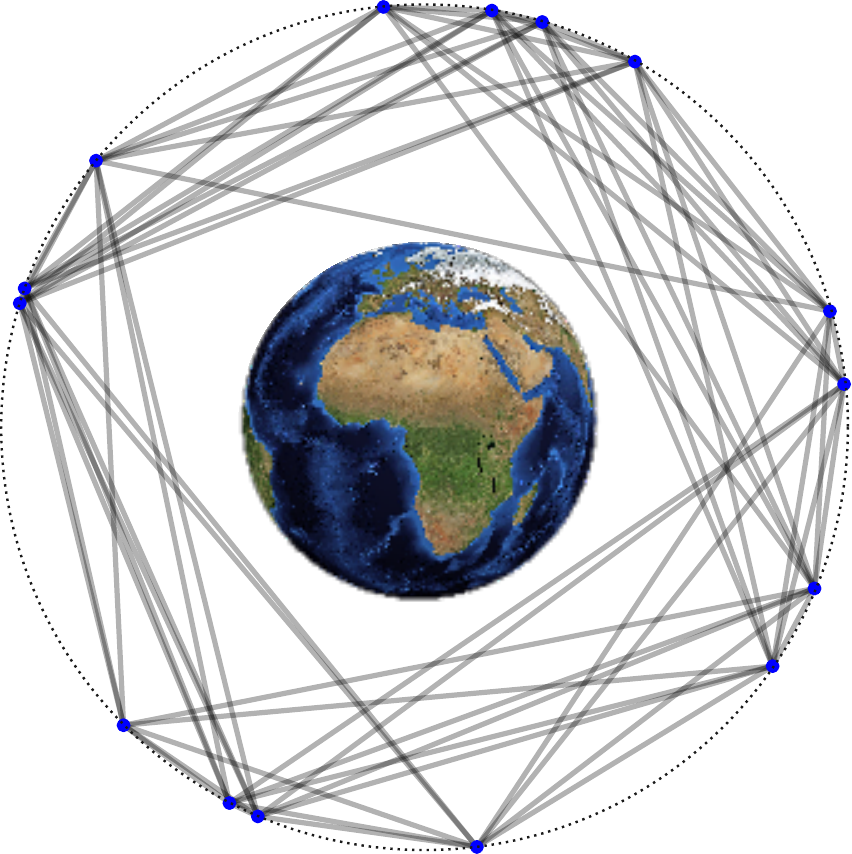}
    \caption{A celestial instance.}
    \label{fig:preliminaries:instanceexamplesB}
  \end{subfigure}\hfill
  \begin{subfigure}[t]{.45\textwidth}
    \centering
    \includegraphics[width=\textwidth]{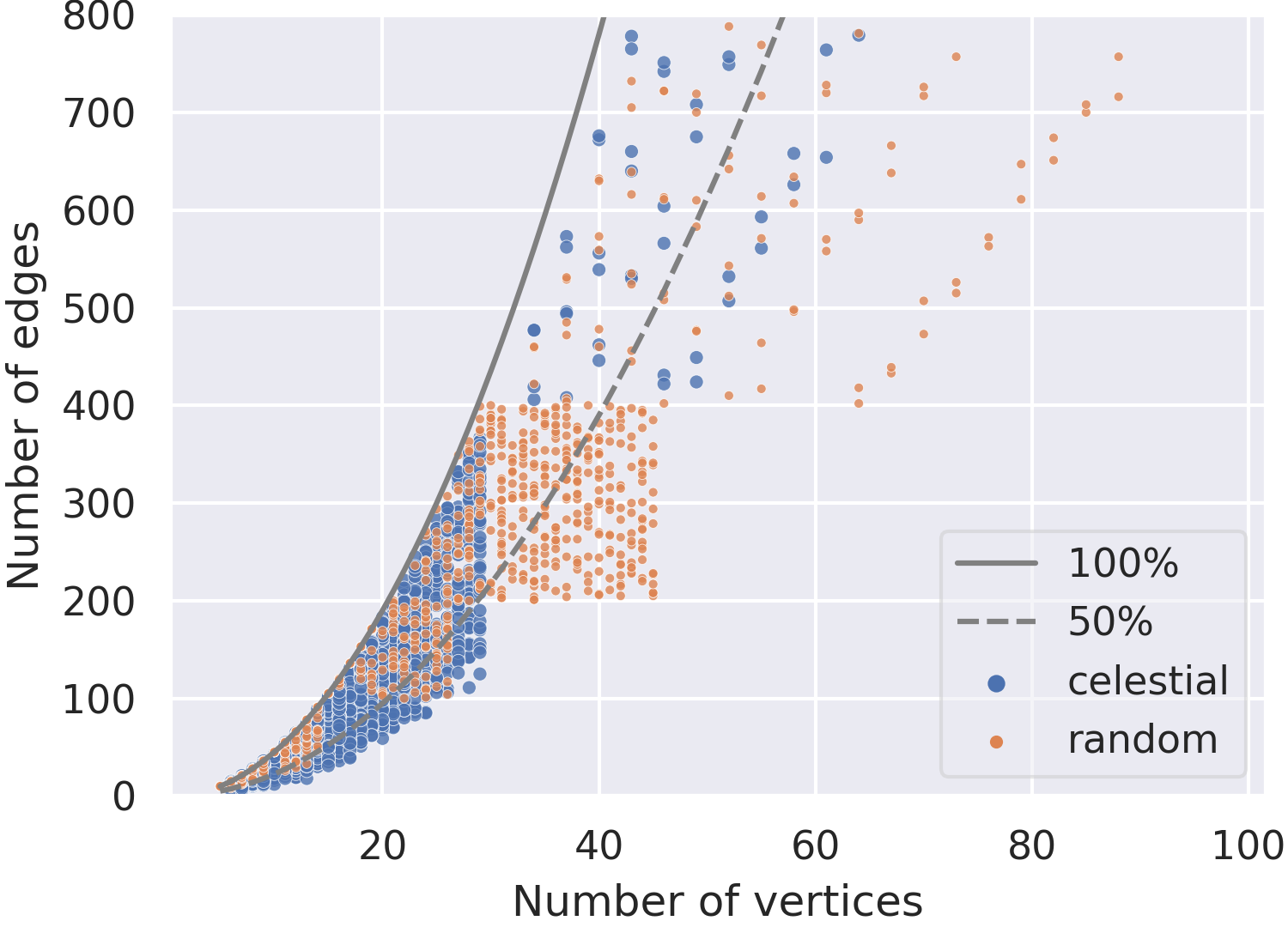}
    \caption{Instance distribution.}
    \label{fig:preliminaries:instanceexamplesC}
  \end{subfigure}\hfil
  \begin{subfigure}[t]{.25\textwidth}
    \centering
    \includegraphics[width=0.95\textwidth]{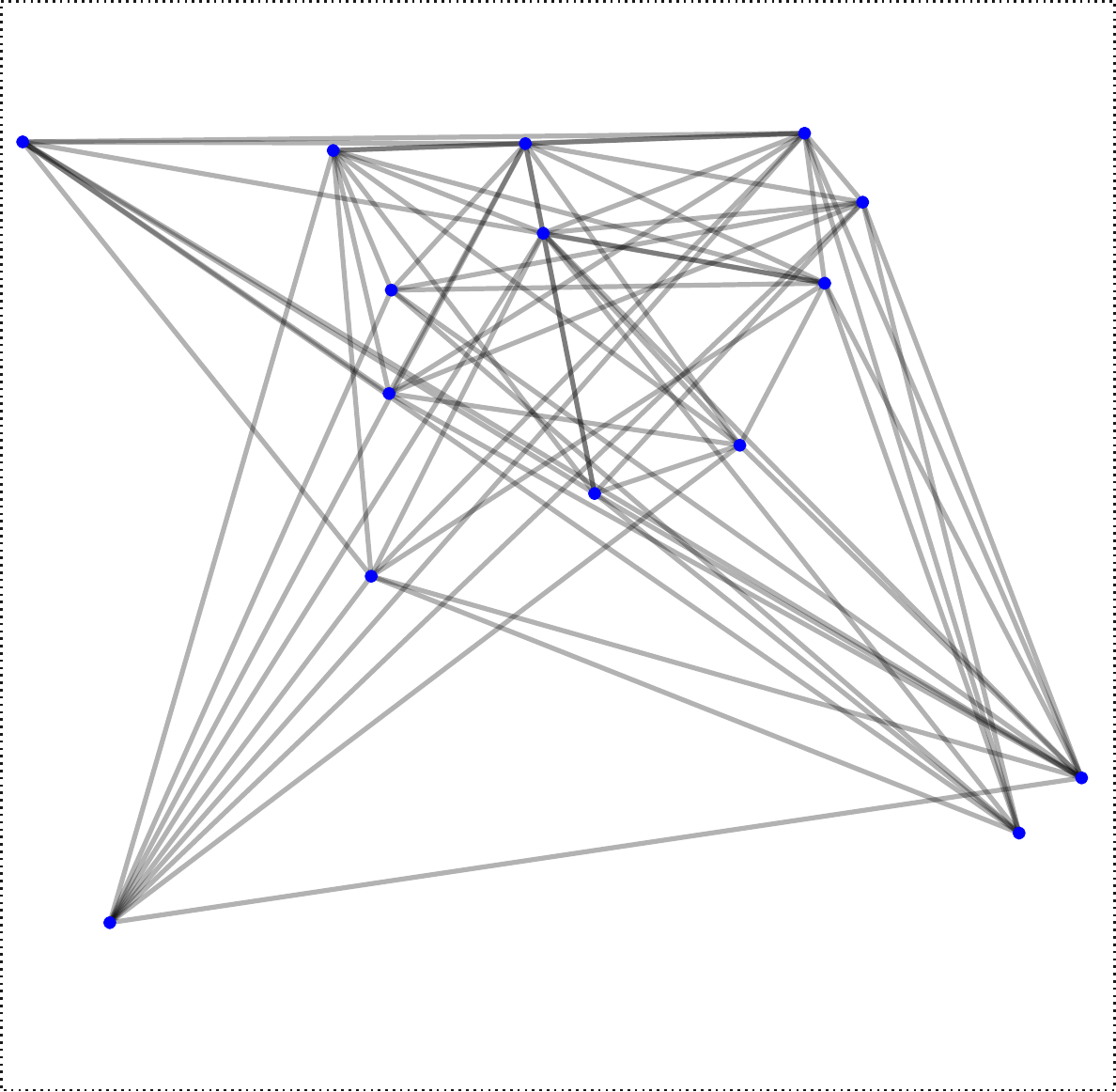}
    \caption{A random instance.}
    \label{fig:preliminaries:instanceexamplesA}
  \end{subfigure}\hfill

  \caption{Examples of instances with 15 vertices and $\sim$ 70 edges: (a) celestial, (c) random, and (b) the instance distribution. Auxiliary lines in (b) indicate graphs with edge densities $50\%$ (dashed) and $100\%$ (solid) of complete graphs.}
  \label{fig:preliminaries:instanceexamples}
\end{figure}

\subsection{Exact Algorithms}
\label{sec:practical:exact}

\newcommand{\constfunctionsymbol}{\alpha}
\newcommand{\costfunction}[1]{\constfunctionsymbol(#1)}

We developed three mixed integer programs (MIPs) and two constraint programs (CPs) to solve instances to provable optimality.
Note that not every program solves all three problems.
An experimental evaluation is given at the end of this section.
While we focus on 2-dimensional geometric instances, all formulations are applicable to all metric cost functions.

\subsubsection{Mixed Integer Program 1 (\MSClength, \MSCtotEnergy, \MSClocEnergy)}
\label{sec:exact:ip1}
Our first MIP, denoted by \MIPone, uses two types of variables.
The first type  are real variables $t_e\geq 0$ for all $e\in E$.
The second type are Boolean variables $x_{(e,e')}$ for all ordered edge pairs $(e, e')\in E^2$. In a computed solution, the variables $t_e$ define a scan cover in which $S(e):=t_e$ and the value of $x_{(e,e')}$ corresponds to $\nu_S(e,e')$.
Because $\nu_S(e, e')=0$ if $|e\cap e'|\not = 1$, we directly set $x_{(e,e')}:=0$ in these cases.
Consequently, the objective functions can be expressed by substitution of $S(e)$ with $t_e$ and $\nu_S(e,e')$ with $x_{(e,e')}$.
Note that a $\min$-$\max$ objective can be implemented by a single additional real variable and one additional constraint for each term in the objective.

We introduce a set of constraints to guarantee that the $t$-variables and the $x$-variables arise from a valid scan cover.
Because the angle function $\alpha$ fulfills the triangle inequality, it suffices to ensure the time difference of the $t$-variables for all $x_{(e,e')}=1$.
We know that $M_1:=\log_2 n \cdot \ang{360}$ is an upper bound on the minimal makespan for a graph $G$ in 2D with $n$ vertices~\cite{DBLP:conf/compgeom/FeketeKK20}.
Moreover, a makespan of $M_2:=|E|\cdot\ang{180}$ allows to scan each edge individually, and thus an optimal scan cover of \MSClocEnergy and \MSCtotEnergy can be realized in this makespan.
Therefore, by inserting the correct $M_i$, we can enforce feasible scan times by using the Big-M method.
\begin{equation}
  \label{eq:mip1:mtz}
  \forall v \in V, \forall (e,e')\in E(v)\times E(v), e\not=e' \colon  \quad t_{e'} \geq t_{e} + \costfunction{e,e'}-(1-x_{(e,e')})\cdot M_i.
\end{equation}
This leaves us with ensuring that the $x$-variables correspond to a feasible scan cover.
First, for every vertex $v$, an incident scanned edge $e$ has at most one predecessor edge and one successor edge in the scan order.
\begin{equation}
  \label{eq:ip1:hamil1}
  \forall v\in V , e \in E(v)\colon \sum_{e'\in E(v), e'\not=e}x_{(e,e')}\leq 1\quad \text{and}\quad  \sum_{e'\in E(v), e'\not=e}x_{(e', e)}\leq 1
\end{equation}
Second, the total number of scanned edges at vertex $v$ is $|E(v)|$, i.e., the number of consecutively scanned edge pairs, is $ |E(v)|-1$.
\begin{equation}
  \label{eq:ip1:hamil2}
  \forall v \in V \colon  \sum_{e,e'\in E(v)\times E(v), e\not=e'} x_{(e,e')} = |E(v)|-1
\end{equation}
Together, \cref{eq:ip1:hamil1,eq:ip1:hamil2} enforce that every vertex has exactly one first and one last scanned edge in the induced scan order.
Because \cref{eq:mip1:mtz} enforces that the scan times obey the rotation times, there are no cycles in the sequence defined by $x$ if all angles are positive.
This fact is very similar to the Miller-Tucker-Zemlin formulation of the TSP~\cite{miller1960integer}.
In the presence of $\ang{0}$-angles, we dynamically add the following constraint similar to the Dantzig formulation~\cite{dantzig1954solution} to separate these cycles.
\begin{equation}
  \label{eq:ip1:dantzig}
  \forall v\in V,\forall S\subsetneq E(v), S\not=\emptyset \colon  \quad \sum_{e \in S, e'\in E(v)\setminus S} x_{(e,e')}+x_{(e',e)} \geq 1
\end{equation}

\subsubsection{Mixed Integer Program 2 (\MSClength)}\label{sec:exact:ip2}

The abstract definition of the MSC~\cite{DBLP:conf/compgeom/FeketeKK20} can be directly implemented as a MIP, because absolute values can be implemented using a Boolean variable.
Some modern solvers like Gurobi actually provide this functionality directly.
Like for \MIPone (\cref{sec:exact:ip1}), we have a real-valued variable $t_e\geq 0$ for each $e \in E$ that states its scan time.
We try to keep the maximum value assigned to any $t_{e}, e\in E$ as low as possible.
For every two incident edges $vw$ and $vu$, we only have the constraint that $t_{vw}$ and $t_{vu}$ have to be at least the time apart that $v$ needs to rotate between these two.
This results in the following \MIPtwo.
\begin{align}
  &\min & \max_{e \in E} t_e&& \\
  &\text{s.t. } &|t_{vw}-t_{vu}|&\geq \costfunction{vu,vw} \quad& \forall vw,vu \in E\\
  & &t_{e} &\geq 0 \quad &\forall e\in E
\end{align}

The main difference to \MIPone is that we do not keep a record of the actually performed rotations.
As a consequence, \MIPtwo can only be used for \MSClength.
However, on the positive side, we do not need to dynamically add additional cycle constraints.

\subsubsection{Mixed Integer Program 3 (\MSCtotEnergy, \MSClocEnergy)}
\label{sec:exact:ip3}

The third MIP
(defined by \cref{eq:ip1:hamil1,eq:ip1:hamil2,eq:ip1:dantzig,eq:ip3:globaldirectedcycle}),
denoted by \MIPthree, is a variant of \MIPone~(\cref{sec:exact:ip1}) in which
the $t$-variables and the corresponding Big-M based constraint (\cref{eq:mip1:mtz}) are removed.
As a consequence,  we may use it  for \MSClocEnergy and \MSCtotEnergy, as they only
need the $x$-variables.

It is possible that the scan orders at the individual vertices are cycle free,
but that the overall schedule has a deadlock when the vertices wait for each
other, see \cref{fig:G8-2-OrderAndDepA,fig:G8-2-OrderAndDepB}.
\begin{figure}[b]
  \centering
  \begin{subfigure}[t]{.3\textwidth}
    \centering
    \includegraphics[scale=1]{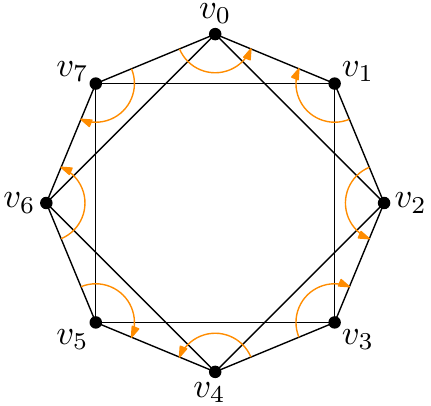}
    \caption{Rotation scheme}
    \label{fig:G8-2-OrderAndDepA}
  \end{subfigure}
  \hfil
  \begin{subfigure}[t]{.3\textwidth}
    \centering
    \includegraphics[scale=1,page=3]{figures/G-8-2-instance-dep.pdf}
    \caption{Cycle in scan order}
    \label{fig:G8-2-OrderAndDepB}
  \end{subfigure}

  \caption{A globally infeasible edge order fulfilling  \cref{eq:ip1:hamil1,eq:ip1:hamil2}, i.e., it is cycle-free at each vertex:
  (a) its rotations scheme
  (b) the resulting edge order that contains a cycle. An arc $(e,e')$ in this graph corresponds to an $x_{(e,e')}=1$.}
\end{figure}
We therefore prohibit directed cycles in the scan order defined by the $x$-variables (if not already
separated by \cref{eq:ip1:dantzig}) dynamically via callbacks for every newly found integral solution.
Violated constraints can be found via a simple DFS search.
\begin{equation}
  \label{eq:ip3:globaldirectedcycle}
  \forall k\in \mathbb{N}_{|V|}, \forall (e_0,e_1, \ldots, e_{k-1})\in E^k \colon \quad x_{(e_{k-1}, e_0)}+ \sum_{i = 0, 1,\ldots k-2} x_{(e_i, e_{i+1})} \leq k-1
\end{equation}
Note that these cycles can also happen in \MIPone, but only with zero rotation costs between the involved edges.
Thus, they are irrelevant for the solution, as all of these edges can be scanned at once.

\subsubsection{Constraint Program 1 (\MSClength)}

Our first constraint program (denoted by \CPone) has the same formulation as \MIPtwo.
The only difference between the CP version and the MIP version lies in the employed solver.
In particular, absolute values can be modeled directly.

\subsubsection{Constraint Program 2 (\MSCtotEnergy, \MSClocEnergy)}

Our second constraint program (defined by \cref{eq:ip1:hamil1,eq:ip1:hamil2,eq:exact:cp2:mtz}), denoted by \CPtwo,
is similar to \MIPthree described in \cref{sec:exact:ip3}.
However, \MIPthree adds \cref{eq:ip1:dantzig,eq:ip3:globaldirectedcycle} dynamically, which our CP does not support.
Because adding all these constraints directly results in a prohibitively large formulation,
we instead use a conditional variant of the Miller-Tucker-Zemlin~\cite{miller1960integer} formulation to eliminate cycles in the scan order.
Different from MIPs, we do not need the Big M method for CPs, but can implement conditional constraints directly.
More precisely, we add the variables $o_e\in \mathbb{N}_{|E|}, e\in E$ that state the cycle-free scan order of the edges, which is enforced by the constraints
\begin{equation}
  \label{eq:exact:cp2:mtz}
  \forall (e,e')\in E\times E \colon \qquad o_{e'}-o_e \geq 1 \quad \text{ if } x_{(e, e')}=1.
\end{equation}

\subsubsection{Experimental Evaluation of Exact Algorithms}
\label{sec:exact:experiments}

We used \emph{Gurobi}~(v9.0.1) for solving the MIPs and \emph{CP-SAT} of Google's \emph{or-tools} (v7.7.7810) for solving the CPs.
CP-SAT, which is based on a SAT solver, requires all coefficients and variables to be integral for computational efficiency.
We therefore convert the floating point values to integral values including the first eight floating point digits (rounded, decimal).
While this weakens the accuracy, we calculated a theoretical maximal deviation of less than $\SI{1e-4}{\percent}$, which we consider negligible and comparable to the accuracy of the MIP solver.

We considered all solvers for the three objectives on the two instance types described in the preliminaries.
We evaluated how many instances of which size could still be solved to provable optimality within a time limit of \SI{900}{\sec};
see \cref{fig:engineering:exact:percentage}.
For \MSClength, \CPone has a clear lead, solving \SI{50}{\percent} of the
instances with $242\pm 5\%$ edges for random instances, and $125\pm 5\%$ edges for celestial instances.
In our experiments, neither MIPs was able to solve any instance with more than \num{70} edges to provable optimality.
For \MSCtotEnergy, \MIPone and \MIPthree performed better than \CPtwo, but all solvers could barely solve instances with more than \num{30} edges.
While \MIPone has a more direct objective without auxiliary constraints and variables as needed for \MSClength, its actual performance was slightly worse.
For \MSClocEnergy, \CPtwo performed considerably better; for celestial instances, it can solve instances nearly twice as large
($\geq \SI{50}{\percent}$ at $48\pm 5\%$ edges) than the MIPs.
Surprisingly, \MIPone was slightly better than \CPtwo for random instances, being able to solve \SI{50}{\percent}
of the instances with $61\pm 5\%$ edges.
Overall, CPs appear to be considerably more effective than MIPs, and random instances show to be easier to solve than celestial ones.

\begin{figure}[htb]
  \newcommand{\h}{5cm}
  \centering
  \begin{subfigure}[t]{.9\textwidth}
    \centering
    \includegraphics[height=\h]{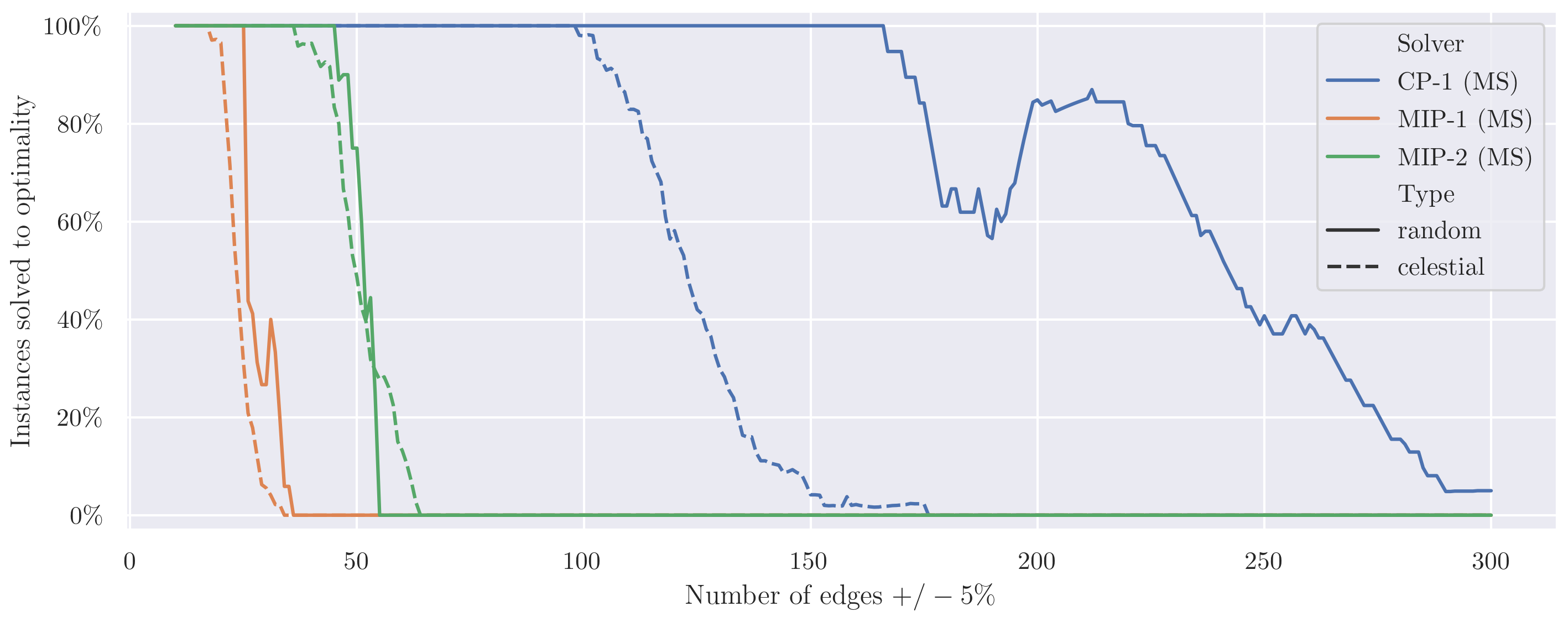}
    \caption{Makespan}
  \end{subfigure}
  \hfill
  
  \begin{subfigure}[t]{.45\textwidth}
    \centering
    \includegraphics[height=\h]{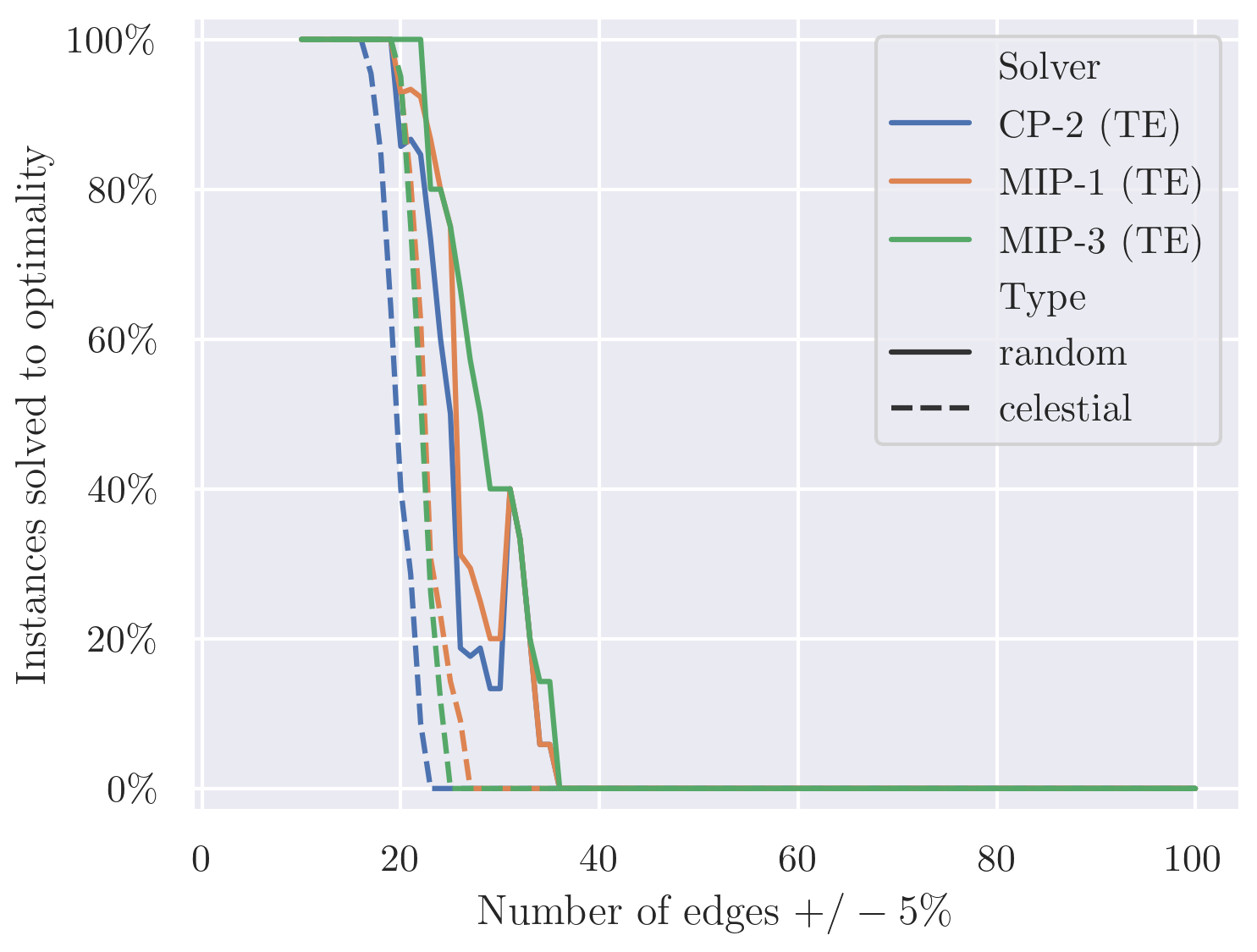}
    \caption{Total Energy}
  \end{subfigure}
  ~ 
  \begin{subfigure}[t]{.45\textwidth}
    \centering
    \includegraphics[height=\h]{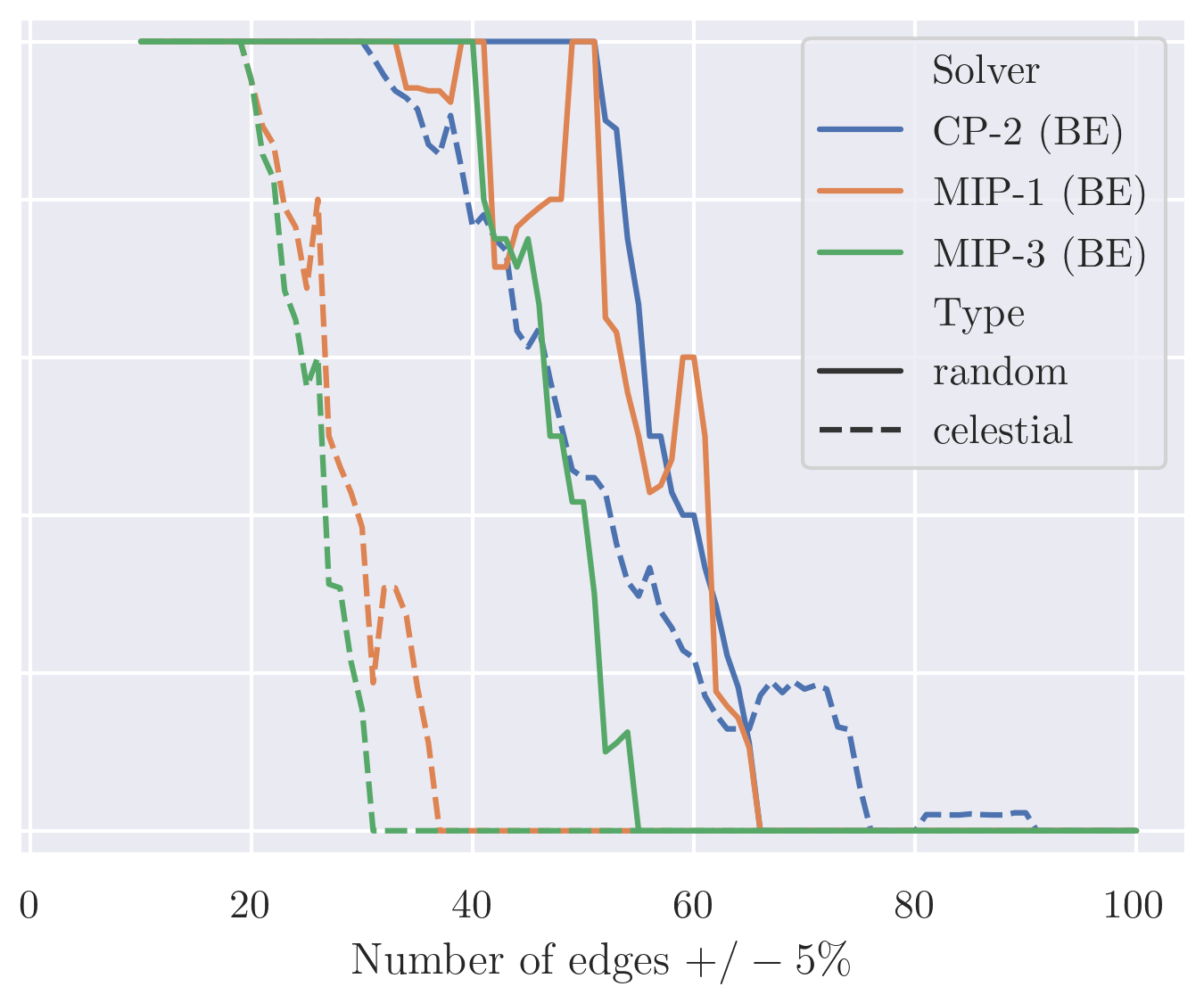}
    \caption{Bottleneck Energy}
  \end{subfigure}
  \hfill

  \caption{Performance of the exact solver measured in how many instances with $m\pm 5\%$ edges can be solved to provable optimality within \SI{900}{\sec}.
  The bump for \CPone starting at \num{200} can be explained by the instance distribution that at this point includes more instances with lower degree.
  }
  \label{fig:engineering:exact:percentage}
\end{figure}

\subsection{Approximations and Heuristics}
\label{sec:practical:heuristic}

For larger instances (beyond the size that was solvable to provable optimality), we developed additional
methods based on approximation algorithms and heuristics that provided good (but not provably optimal) solutions.

\subsubsection{Bipartite Approximation Algorithms with Coloring Partition}

The constant-factor approximation algorithms for bipartite graphs extend to general graphs by partitioning them into bipartite graphs and applying  the corresponding approximation algorithm to each of the bipartite subgraphs.
	More specifically, 
	assigning a vector over $\{0,1\}$ with $\lceil \log_2 k\rceil$ bits to each color class of a $k$-colored graph induces a covering of its edge set with $\lceil \log_2 k \rceil$ bipartite graphs; for more details see  Motwani and Naor~\cite{Motwani:1994:EAC:891890}.
For \MSClength, this even preserves the approximation factor~\cite{DBLP:conf/compgeom/FeketeKK20}.
We use the well-engineered \emph{dsatur} heuristic~\cite{brelaz1979new} for the graph coloring problem, which is shipped with
the \emph{pyclustering}-package~\cite{Novikov2019}.
Concatenating the solutions of the bipartite graphs yields a feasible scan cover; here we use a greedy approach to minimize the transition costs.
We denote this method by APX.

\subsubsection{(Meta-)Heuristics}
We also considered a number of (meta-)heuristics for optimizing the three objectives.
\begin{description}
  \item[Greedy:] Scan the first edge regarding a given or random order and then scan the edge that increases the objective the least, until all edges are scanned.
    If multiple edges are equally good, the first one regarding the order is selected. Many edges can be inserted without extra cost and thus
    the initial edge order has a strong influence on the result.
  \item[Iterated Local Search (ILS):] This simple but potentially slow heuristic considers for a given start
    solution (in this case of \emph{Greedy}) all possible swaps of edges; the locally best swap is carried out, until no further improvement is possible.
  \item[Simulated Annealing (SA):] This common variation of \emph{Iterated
    Local Search} performs swaps according to a probability based on the Boltzmann
    function $\textbf{Boltzmann}(T, s_1, s_2)=e^{\nicefrac{1}{T}\cdot(s_2-s_1)}$, where $s_1$ is
    the objective value of the current best solution, $s_2$ is the objective value
    of the considered solution, and $T\in \mathbb{R}^+$ is the current \emph{temperature}.
    The temperature decreases over time and with it the likelihood of a worse
    solution being used.
    If the objective does not improve for some time, the temperature is increased in order to escape the local minimum.
    Due to randomization, we can run multiple searches in parallel. We terminate if the solution has not improved for some time.
  \item[Genetic Algorithm (GA):] We start with an initial population of \num{200} solutions generated by a randomized Greedy.
    A solution is encoded by assigning each edge a fractional number between $0.0$ and $1.0$, similar to~\cite{gholami2009scheduling}.
    The scan order is determined by sorting the edges by these numbers.
    In each round, we build a new population by selecting the best $10\%$ of
    the old population (\emph{elitism}) and then fill the rest of the population by
    crossovers of the old generation.  For a crossover, we select two solutions of
    the old generation with a probability matching their objective values
    (\emph{uniform selection}) and for each edge we choose with equal probability
    either the number from the first or second solution (\emph{uniform crossover}).
    If by chance, two edges get the same number, we randomly change one of them without influencing the order.
    Of the new generation of solutions, $3\%$ are selected for mutation.
    A mutation applies Greedy with a probability of $60\%$ (the old order is used as initial edge order) or changes
    each edge with a $3\%$ probability to a new random number.
    This is repeated until we either reach a time limit of $\SI{900}{\sec}$, \num{300} generations, or \num{60} generations without improvement.
    The best solution found during this process is then returned.
\end{description}

\subsubsection{Experimental Evaluation of Approximations and Heuristics}

\cref{fig:engineering:inexact:overview} shows experimental results for heuristically solving instances with up to \num{800} edges with a \SI{900}{\sec} time limit (at which point the current solution is returned).
For \MSClength, \CPone yields the best results even for larger instances (where it is aborted by the time limit) by a margin of \SIrange{25}{50}{\percent} to the next best algorithm, GA\@.
For \MSCtotEnergy, the genetic algorithm turned out to be the best approach for celestial instances by a margin
of over \SI{50}{\percent} for the larger instances.
Surprisingly, \CPone (optimizing for \MSClength) yields slightly better solutions than the genetic algorithm for random instances of \MSCtotEnergy.
The most interesting results are for \MSClocEnergy.
Here, \CPone achieves the best results by a margin of over \SI{20}{\percent} for random instances, and GA (TE) the best results for celestial instances by a margin of over \SI{40}{\percent}.
The excellent performance of \CPone can be explained by a strong correlation of
\MSClength and \MSClocEnergy for random graphs, as shown in
\cref{fig:engineering:correlations}.  The fact that GA (TE) is actually better
in optimizing \MSClocEnergy than GA (BE) can be explained by the weaker
gradients of bottleneck objectives, because only a small part of the solution
(the most expensive vertex) actually contributes to the value.
However, the initial bump, at which the exact solver of \MSClocEnergy still yields
(better) solutions, indicates that these solutions could be far from optimal
and that there may still be room for improvement.

\begin{figure}[htb!]
  \newcommand{\h}{5.5cm}
  \centering
  \begin{subfigure}[t]{.96\textwidth}
    \includegraphics[height=\h]{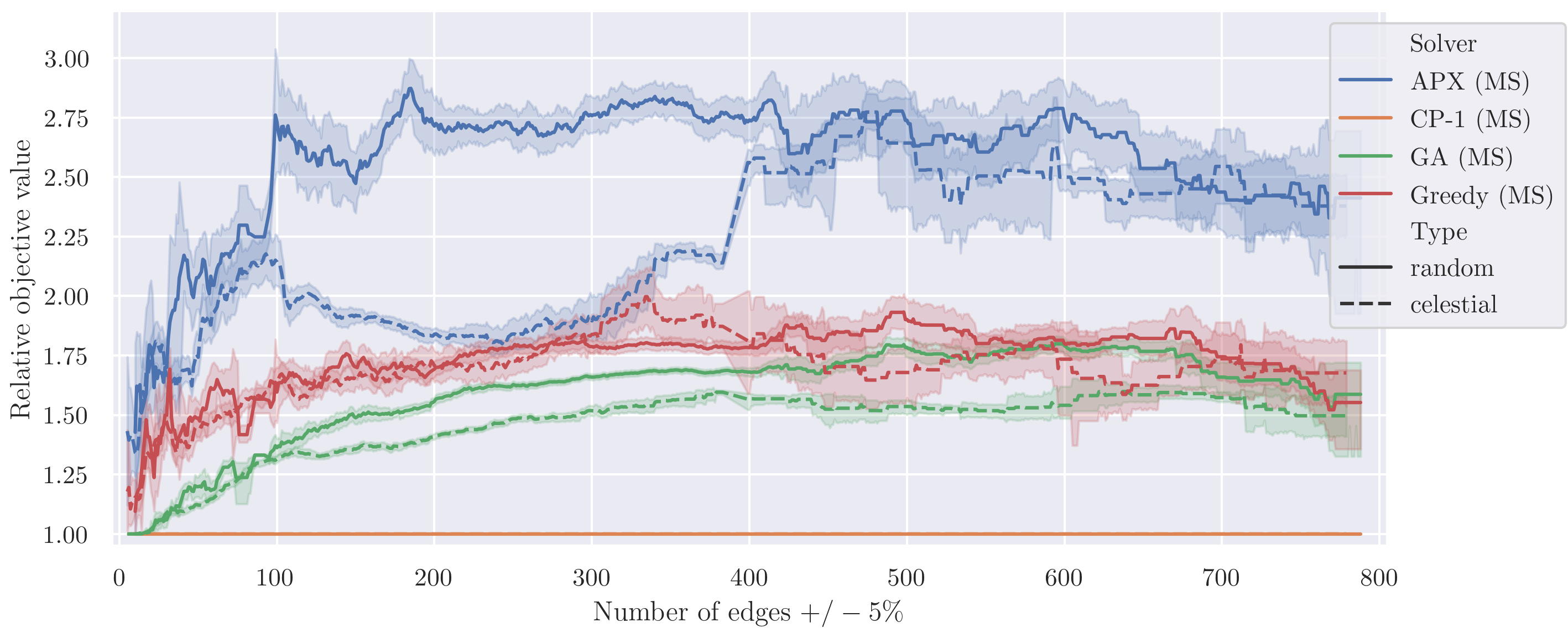}
    \caption{Makespan.}
  \end{subfigure}
  \hfill
  
  \begin{subfigure}[t]{.45\textwidth}
    \includegraphics[height=\h]{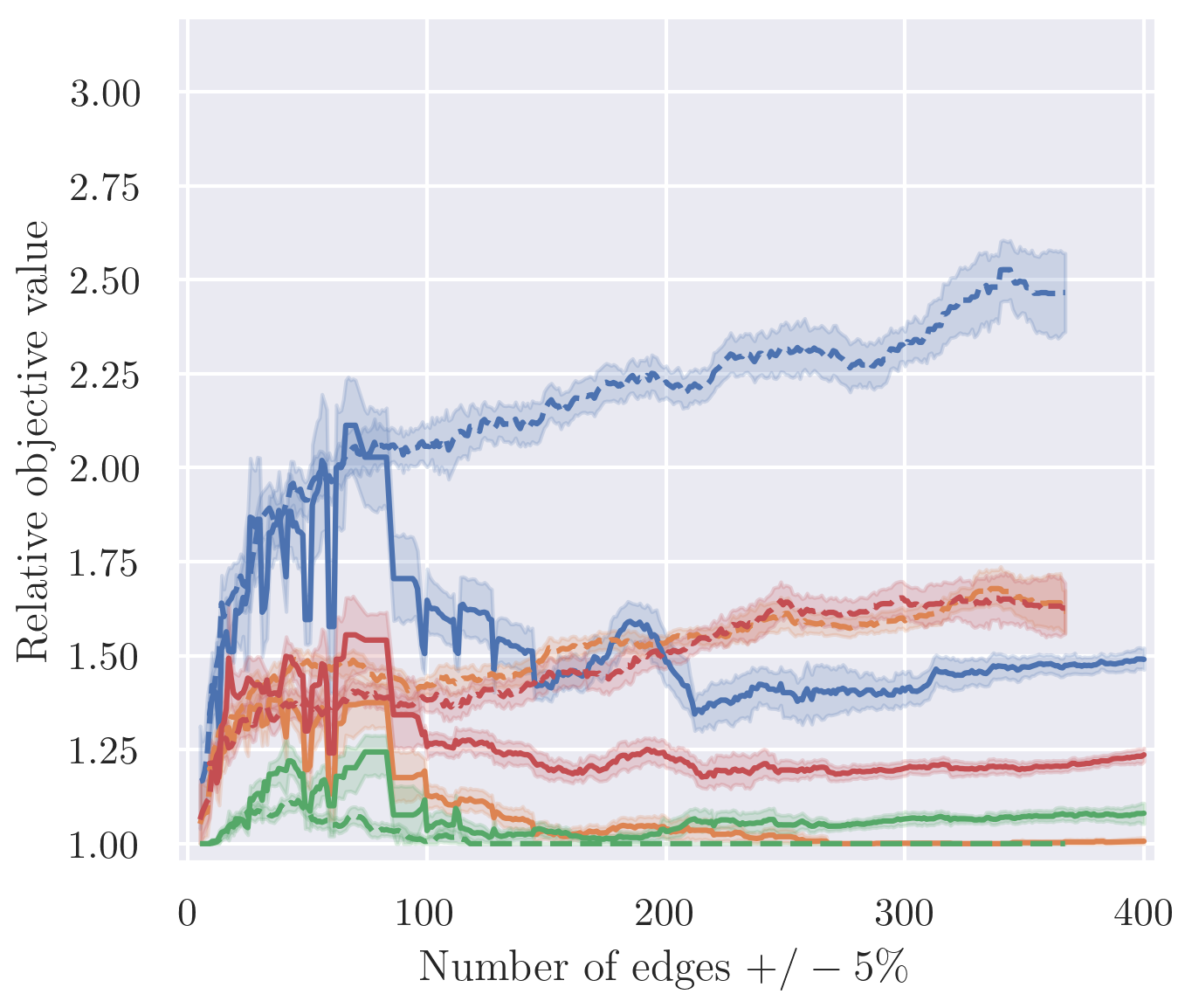}
    \caption{Total energy.}
  \end{subfigure}
  \hfil
  \begin{subfigure}[t]{.52\textwidth}
    \centering
    \includegraphics[height=\h]{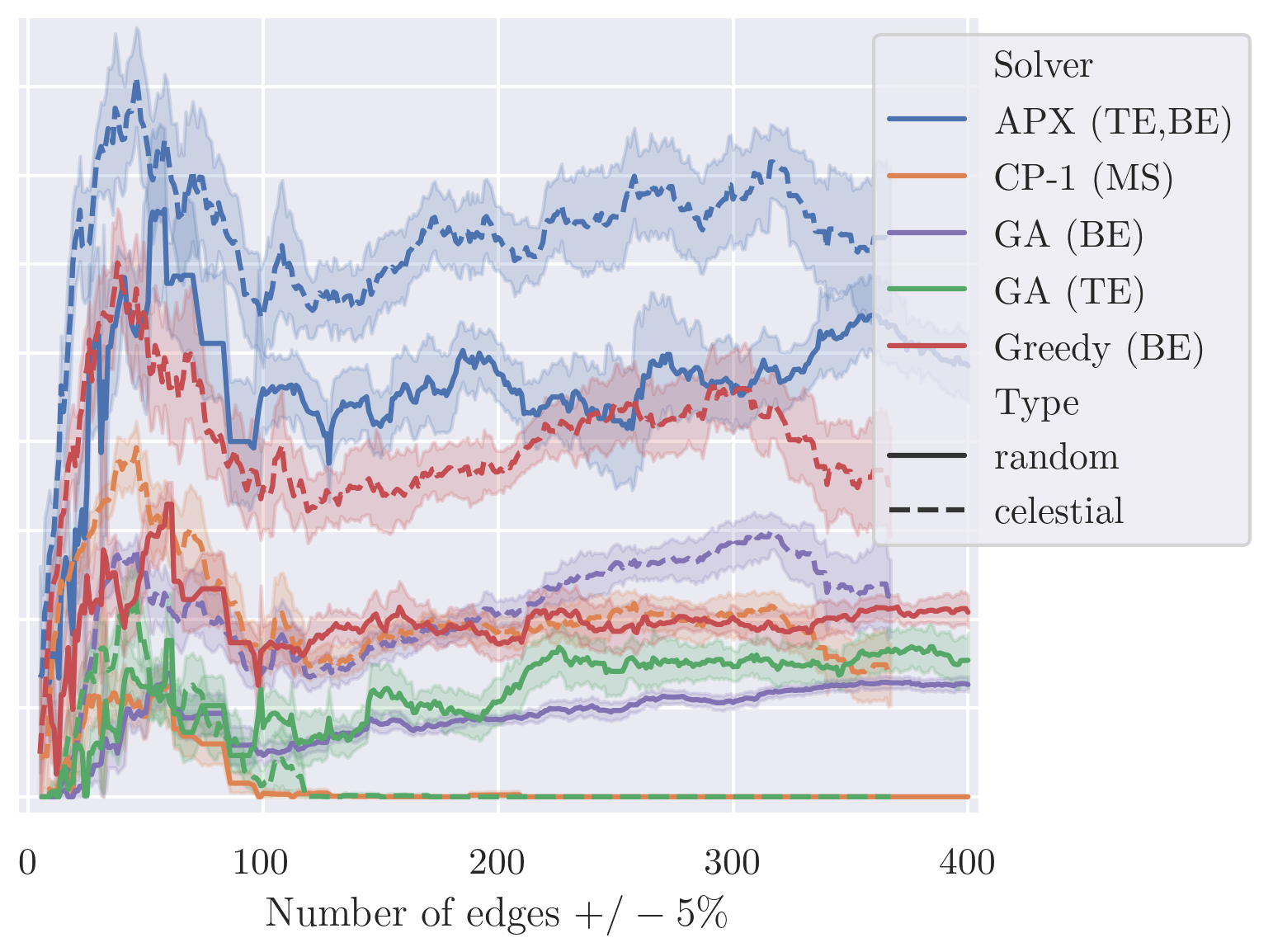}
    \caption{Bottleneck energy.}
  \end{subfigure}

  \caption{%
    Relative performance of the non-exact methods, measured by the obtained objective value divided by the best known value.
    We used the same instances for the exact solver, so the better denominator creates a small bump for smaller instance sizes, 
    in particular for \MSClocEnergy.
    Except for \CPone, the exact solvers did not yield good solutions for larger instances, if any at all, and are thus excluded for readability.
    The plots show the mean and the corresponding \SI{95}{\percent} confidence interval.
    We highlight the difference between the two instance types by using different styles for the lines.
    Note that because these are relative values, a comparison of the performance over the different objectives is not possible.
    ILS and SA are excluded for readability and perform only slightly better than Greedy.}
    \label{fig:engineering:inexact:overview}
\end{figure}

Overall, either \CPone or GA (TE) yields the best solutions.
\CPone is especially strong on random instances for all three objectives.
The approximation algorithm is usually among the worst.
For \MSClength, the algorithm performs a full rotation for nearly all instances, as $\max_{v\in V} \Lambda(v)$ is usually above \ang{180}.
Note that the factor can be worse than the approximation factor $4.5$ (resp.\ $2$), because these are not bipartite graphs.

In \cref{fig:engineering:correlations}~(first row, fourth and last column) we can additionally see that for \MSClength the objective correlates strongly with the number of edges for celestial instances and with the average degree for random instances.
Total energy seems to primarily correlate with the number of edges for both types; our random instances are on average twice as expensive.
For \MSClocEnergy, only random graphs seem to have a significant correlation to \MSClength and the average degree.

\begin{figure}[htb!]
  \centering
  \includegraphics[width=1.1\textwidth]{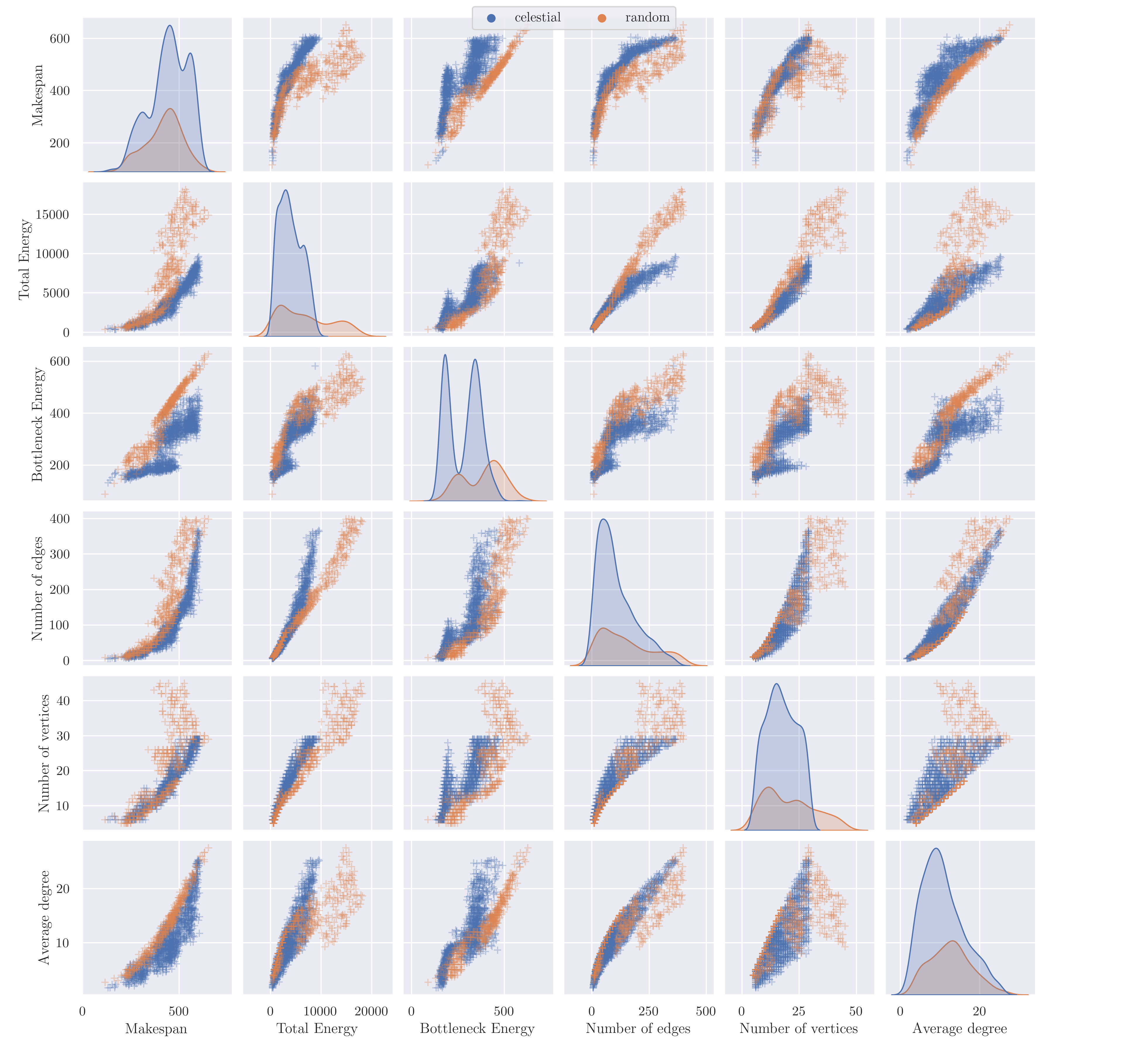}
  \caption{%
    Correlation and distribution of the best known objectives and instance properties.
    The diagonal shows the density distribution of the x-values.
    The scatter plots have a point for every existing value pair, which allows to detect correlations.
    }
    \label{fig:engineering:correlations}
\end{figure}

\section{Conclusion and Open Problems}

We studied problems of minimum scan cover with three different practically relevant objective functions,
providing both theoretical and practical contributions:
complexity and algorithmic results for the new objectives (\MSCtotEnergy and \MSClocEnergy),
and practical methods for computing provably optimal solutions for smaller and near-optimal solutions for larger instances.

In particular, we developed multiple MIP and CP formulations and demonstrated that instances of \MSClength can be solved reliably 
for instances with more than \num{100} edges using constraint programming which performs much better than our MIP approaches.
While this approach generalizes also to 3D, we only tested 2D instances; it is open whether these results also carry over to 3D.
\MSCtotEnergy and \MSClocEnergy can only be solved to optimality for much smaller instances.
For solving larger instances without guarantee of optimality, we evaluated approximation algorithms and a spectrum of meta-heuristics.
Within the given time limit, \CPone provided the best solutions for all \MSClength instances,
and even the random instances for \MSCtotEnergy and \MSClocEnergy, despite only optimizing for \MSClength.
For celestial instances of \MSCtotEnergy and \MSClocEnergy, the genetic algorithm optimizing for \MSCtotEnergy provides the best solutions.
However, the results indicate perspectives for improving the optimization of \MSClocEnergy. 

At this point,
fully dynamic instances (in which the vertices change their relative positions to each other over time, such 
as for satellites with different orbit parameters) are yet to be explored. These promise to be even more challenging,
due to bigger gaps between optimal and suboptimal solutions, resulting from possibly long delays when a limited communication window
has been missed.



\bibliography{biblio_min.bib}

\appendix
\section{\MSCtotEnergy and \MSClocEnergy in 1D}
\label{appendix:1d}

In this section we present a complete proof of the following theorem.

\onedsimple*
\begin{proof}
	We assume that the vertices are placed on a horizontal line.
	We partition the vertices into two groups: those with a neighbor to only one side and those with neighbors to both sides.
	If the second group is empty, there is a trivial zero-cost solution for both objectives.

	Thus, consider $k>0$.
	Each of these vertices needs to rotate at least $\ang{180}$, so the
values  $ \ang{180}\cdot k$ and $\ang{180}$ are lower bounds for \MSCtotEnergy
and \MSClocEnergy, respectively.  The following strategy matches this lower
bound:
	The vertices in the first group are headed to their neighbor and do not rotate. In the following we restrict our attention to the vertices in the second group.
	In the beginning, all of them are headed left.
	Then, from left to right, one after the other rotates such that it is headed right; the next vertex starts only after the completion of its predecessor. Note that whenever a vertex rotates, all edges to its left are scanned. Consequently, this yields a valid scan cover.
\end{proof}

\section{NP-hardness of \MSClocEnergy and \MSCtotEnergy} 
\label{appendix:nphardness}
Given an instance $I$ of \textsc{MNAE3SAT}, we construct a graph $G_I$ with the same $\Lambda(v)$ for all vertices that has a \mincover if and only if $I$ is satisfiable.
Recall that the \textsc{Minimum Scan Cover} problem with respect to the objectives \MSClocEnergy and \MSCtotEnergy on $G_I$ is equivalent to finding a \mincover of this graph.
The choice of variable assignment is encoded by the choice of rotation direction in a \mincover of specific vertices in $G_I$, which we call \emph{connector vertices}.

\subparagraph{Constructing the gadgets}
We construct a \textit{variable gadget} $G_x$ (\cref{fig:variable-gadget}) for
each variable $x \in X$ and a \textit{clause gadget} $G_C$
(\cref{fig:clause-gadget}) for each clause $C \in \mathcal{C}$. For  $x \in C$,
we connect the gadgets $G_x$ and $G_C$ with a  \textit{wire gadget} $G_w$
(\cref{fig:wire-gadget}). The resulting graph $G_I$ is symbolically shown in
\cref{fig:hardness-construction-example}. We construct both the variable gadget
and wire gadget from smaller components called \textit{wire fragments} $G_f$,
see \cref{fig:swap-fragment} for an illustration.

\bigskip

We first state several observations that help with the construction of these gadgets.

\begin{observation}\label{obs:min-scan-first-last}
	If an edge $vw$ is the first or the last edge scanned in a \mincover, it bounds a minimal $\Lambda$-\cone of vertex $v$ and of vertex $w$.
\end{observation}

In the gadgets, we will prescribe the two edges bounding the $\Lambda$-\cone of a vertex.

\begin{observation} \label{clm:lambdacone}
Consider a straight-line drawing of a graph $G$. For every vertex $v$ and every pair of  consecutive edges $e,e'$ at $v$,  we can add edges incident to $v$ (and new vertices of degree 1), such that in the resulting drawing,  $e$ and $e'$ bound the $\Lambda$-cone of $v$.
\end{observation}
\cref{clm:lambdacone} allows us to choose the maximal angle between consecutive
edges. In the same manner, we can ensure that $\Lambda(v)$ is equal for all
vertices $v$ of the gadgets (excluding the newly added vertices of degree 1).

Next we construct the individual gadgets. Consider a  wire fragment $G_f$ as depicted in \cref{fig:swap-fragment}.
Using \cref{clm:lambdacone}, we make sure that the $\Lambda$-cones correspond to the blue arcs.
The vertices $s,t,u$ can be shared between subgraphs with the following properties concerning their direction of rotation in a \mincover. Note that given a \mincover, we can obtain a second \mincover by reversing all directions.

\begin{lemma}\label{lem:swap-properties} The vertices $s,t,u$ in the wire fragment $G_f$ have the following properties.
\begin{enumerate}
		\item In every \mincover of $G_f$, the vertices $u$ and $t$ rotate in the same direction, while $s$ rotates in the opposite direction.
		\item There exists a \mincover of $G_f$.
	\end{enumerate}
\end{lemma}
\begin{proof}
	Suppose we have a \mincover of $G_f$ with scan order $P$. Note that one of the edges $s u$, $t v_2$ is first in $P$, and the other last (these are the only edges bounding a minimal $\Lambda$-cone, see Observation~\ref{obs:min-scan-first-last}). Suppose $s u$ is the first edge in $P$. Then $s$ turns counterclockwise and $u$ turns clockwise. Additionally, $t$ turns clockwise, because $t v_2$ is the last edge. The other case, in which $t v_2$ is the first edge, is analogous.
	
	The following scan order yields a \mincover of $G_f$ in which $s$ rotates counterclockwise, see also the edge labels in Figure~\ref{fig:swap-fragment}: $s u$, $u v_3$, $v_2 v_3$, $v_1 v_2$, $u v_1$, $v_1 v_5$, $s v_4$, $s v_5$, $t v_5$, $v_3 v_5$, $v_3 v_4$, $v_1 v_4$, $t v_4$, $t v_2$. 
\end{proof}

\begin{figure}[htb]
	\centering
	\begin{subfigure}[t]{.48\textwidth}
		\centering
    \includegraphics[scale=.9]{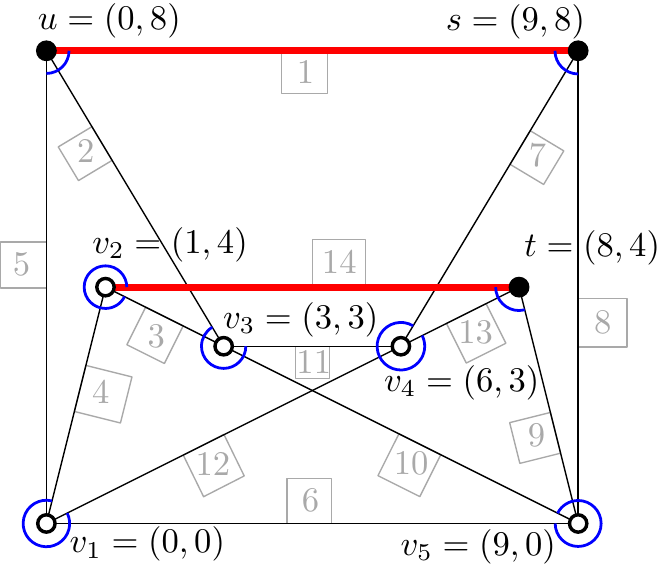}
		\caption{Wire fragment with vertices $s,t,u$ as connectors. Coordinates are given relative to $v_1$.
		}
		\label{fig:swap-fragment}
	\end{subfigure}\hfill
	\begin{subfigure}[t]{.48\textwidth}
		\centering
      \includegraphics{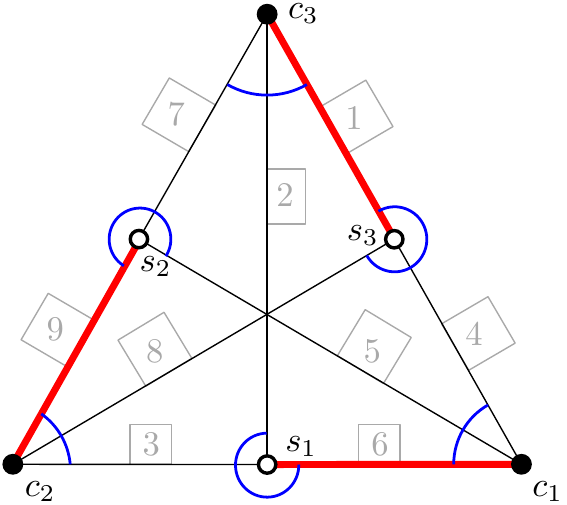}
	    \caption{Clause gadget with vertices $c_1,c_2,c_3$ as connectors on the corners of an equilateral triangle.}
	\label{fig:clause-gadget}
	\end{subfigure}
	\caption{Blue arcs indicate the $\Lambda$-\cone of each vertex. Red edges are candidates for the first or last scanned edge in a \mincover. Grey edge labels indicate a scan order of a \mincover.}
	\label{fig:gadgets}
\end{figure}

For a variable $x\in X$, the variable gadget $G_x$ consists of a chain of wire fragments, as depicted in \cref{fig:variable-gadget}. Denoting the number of occurrences of $x$ in $I$ by $k$, we create $2k$ copies of the wire fragment $G_f^i$ with vertices $s_i,t_i,u_i$. Rotate the wire fragments with even index by $\ang{180}$.
To combine the wire fragments, we  identify the vertices $s_i$ and $s_{i-1}$ for even $i$ and
the vertices $u_i$ and $u_{i-1}$ for odd  $i$.
We define $V_x^c:=\{t_{2i}\mid i=1,\dots, k\}$ as the \emph{connector vertices} of the variable gadget.

\begin{figure}[htb]
	\centering
	\includegraphics{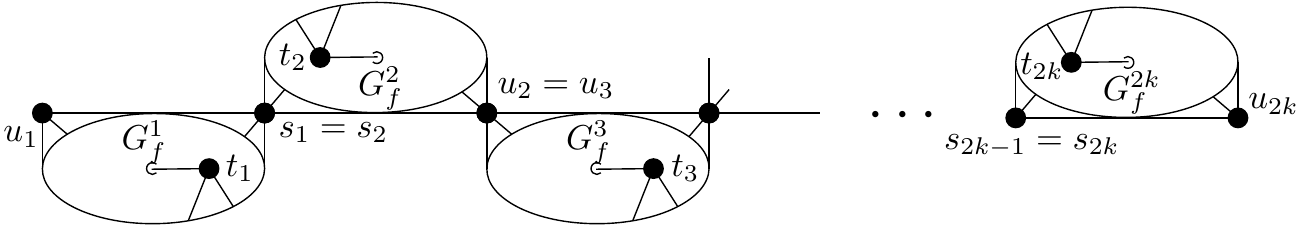}
	\caption{The variable gadget $G_x$ of a variable $x\in X$ consists of $2k$ wire fragments, with every second copy rotated by \ang{180}.}
	\label{fig:variable-gadget}
\end{figure}

\begin{lemma}\label{lem:variable-gadget}
	The variable gadget $G_x$ has the following properties.
	\begin{enumerate}
		\item In every \mincover of $G_x$, all vertices in $V_x^c$ rotate in the same direction.
		\item There exists a \mincover of $G_x$.
	\end{enumerate}
\end{lemma}
\begin{proof}
	Suppose we have a \mincover of $G_x$. By \cref{lem:swap-properties}, $u_i$ and $t_i$ rotate in the same direction, while $s_i$ rotates in the opposite direction. Because the wire fragments are all connected at vertices with identical properties, all copies rotate in the same direction across all wire fragments; in particular, the vertices $t_i$ rotate in the same direction.
	
	For each wire fragment $G_f^i$ in $G_x$, there exists a scan cover that
is a \mincover for the fragment (\cref{lem:swap-properties}). Concatenating the
schedules from $G_f^{1}$ to $G_f^{k}$ in this order yields a schedule for
$G_x$. This results in a \mincover for $G_x$, as vertices used in two
fragments rotate in the same direction and have two $\Lambda$-cones, one for
each scan order of the two fragments.  \end{proof}

For the construction of the clause gadget $G_C$ of $C\in\mathcal{C}$, see
\cref{fig:clause-gadget}. We place the \emph{connector vertices}  $c_1,c_2,c_3$
at the corners of an equilateral triangle, the vertices $s_1,s_2,s_3$ on the
midpoints of the sides as illustrated, and insert the edges $c_i s_j$ for all
$i,j\in \{1,2,3\}$.  Using \cref{clm:lambdacone}, we ensure that the
$\Lambda$-cones of $s_1$, $s_2$, and $s_3$ correspond to the blue arcs in the
figure.
Note that a small perturbation suffices to obtain rational coordinates and does not  harm the construction.

\begin{lemma}\label{lem:clause-gadget}
	The clause gadget $G_c$ has the following properties.

	\begin{enumerate}
		\item In every \mincover of $G_c$, not all connector vertices in $G_c$ rotate in the same direction.
		\item For each assignment of directions to vertices in $G_c$
that does not assign them all the same direction, there exists a \mincover of
$G_c$, such that every vertex in $v$ rotates in its assigned direction.
	\end{enumerate}
\end{lemma}
\begin{proof}
Consider a \mincover $S$ of $G_c$.
Note that a connector vertex $c_i$ turns clockwise in $S$ if and only if it scans the edge $s_ic_i$ first; this edge is highlighted red in \cref{fig:clause-gadget}.

By \cref{obs:min-scan-first-last}, the edges $s_ic_i$ are the unique candidates for the first or last edge scanned in $S$. Consequently,  at least one connector vertex turns clockwise, and at least one turns counterclockwise.

What remains to be shown is that for all configurations of not all equal rotations, there exists a scan order that covers the minimum angle. The order $c_3 s_3$, $c_3 s_1$, $c_2 s_1$, $c_1 s_3$, $c_1 s_2$, $c_1 s_1$, $c_3 s_2$, $c_2 s_3$, $c_2 s_2$ is minimal and has $c_3$ clockwise and $c_1,c_2$ counterclockwise. Reversing this order is also minimal and has $c_1,c_2$ clockwise and $c_3$ counterclockwise. Up to symmetry, these are all possible configurations in which $c_1,c_2,c_3$ do not all rotate in the same direction.
\end{proof}

A wire gadget $G_w$ consists of a chain of $18$ wire fragments $G_f^i$ such
that two connectors on the ends of the chain differ in angle by $\theta$.
Observe that the angle between the bisectors of the maximum angles of vertices
$u_i$ and $s_i$ is $\ang{90}$. The construction consists of five parts
that each will rotate the chain at an angle of $\theta/5$. (We use five parts to ensure the angle $\ang{90}-\theta/5$ is not too small, which we need for \cref{cor:approximtion-hardness-2D}) We connect the first four
fragments, such that $u_2 = s_1$, $s_3 = s_2$, and $s_4=u_3$
(Figure~\ref{fig:wire-gadget}). We match the bisectors of all these
connections, except for the connection between $G_f^1$ and $G_f^2$, which
is $\ang{90}-\theta/5$. This is repeated four more times. The final part has only two wire fragments with an angle $\ang{90}-\theta/5$ between them.

\begin{figure}[htb]
	\centering
	\includegraphics[page=2]{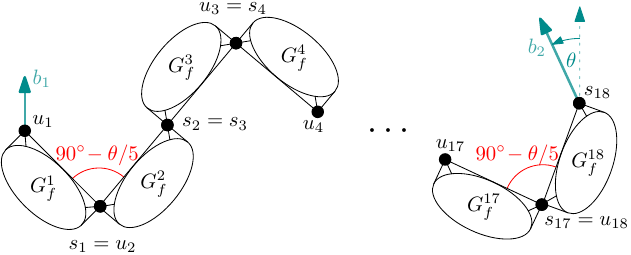}
	\caption{The wire gadget consists of a chain of $18$ wire fragments $G_f^i$. Red arcs indicate angles that vary such that the angle of $b_1$ and $b_2$ is equal to a given $\theta$.}
	\label{fig:wire-gadget}
\end{figure}

\begin{lemma}\label{lem:wire-gadget}
	Let $b_1$ and $b_2$ be the bisectors of the outer cones of $u_1$ and $s_{18}$ in the wire gadget $G_w$, respectively (see \cref{fig:wire-gadget}).
	For any angle $\theta$, the wire gadget $G_w$ can be constructed such that
	the (counterclockwise) angle between $b_1$ and $b_2 $ is $\theta$ and fulfills the following properties.
\begin{enumerate}
		\item In every \mincover of $G_w$, both connector vertices $u_1$ and $s_{18}$ in $G_w$ rotate in the same direction.
		\item There exists a \mincover of $G_w$.
	\end{enumerate}
\end{lemma}
\begin{proof}
	By construction, the angle between the bisectors of the outer cones of $u_{4i+1}$ and $s_{4i+2}$ is $\theta/5$, and the angle between the bisectors of the outer cones of $s_{4i+2}$ and $s_{4(i+1)+1}$ is~$0$. Therefore, the angle between the bisectors of the outer cones of $u_1$ and $s_{18}$ is $5\cdot \theta/5 = \theta$.
	
	Consider the order of the connector vertices of $G_w$, starting at $u_1$, and denote the $i$-th vertex by $c_i$.
	Suppose we have a \mincover of $G_w$. Any pair of connector vertices in the chain that belong to the same wire fragment rotates in the opposite direction (Lemma~\ref{lem:swap-properties}.1). Therefore, connector vertices with an odd number of connector vertices between them on the chain rotate in the same direction. Therefore, $u_1$ and $u_{18}$ rotate in the same direction.
	
	We now give a schedule in which $u_1$ and $s_{18}$ rotate clockwise.
By Lemma~\ref{lem:swap-properties}.2, there exists a \mincover schedule
relative to each wire fragment, such that  $c_i$ rotates clockwise if and only if $i$
is odd. Concatenating the schedules from $G_f^{18}$ to $G_f^1$ in this
order yields a schedule for $G_w$. The vertices internal to the wire fragments
have the same angles as in the earlier cover.  Each connector vertex that is
involved in one of the $\ang{90}-\theta/5$ angles rotates in
counterclockwise direction and has the angle of size
$\ang{90}-\theta/5$ inside the minimum \cone, so the concatenation of
the schedules covers this vertex with a minimum \cone.
	Each connector vertex that is not involved in one of the $\ang{90}-\theta/5$ angles has two $\Lambda$-cones, so these vertices can be scanned in a \mincover for any order of the gadgets.
\end{proof}

Finally, to construct $G_I$, we first place corresponding variable and clause gadgets in the plane. We place the gadgets in a row with all clauses to the right of all variables.
 For a variable that occurs in a clause, we connect their gadgets by a wire gadget. To this end, we first identify a connector vertex $c := t_{2i}$ of the variable $G_x$ and the connector vertex $\overline c := u_1$ of the wire gadget $G_w$.
 We rotate the wire such that the bisector $b$ of the $\Lambda$-cone of $c$ in $G_x$ and the bisector $\overline b $ of the outer cone of $\overline c$ in $G_w$ are equal. This ensures that the identified vertex $c=\overline c$ has two $\Lambda$-cones of size $\Lambda(c)$ (\cref{fig:connecting_gadgets}).

  By \cref{lem:wire-gadget} we can choose the bisector of $s_{18}$ corresponding to the connector vertex $c'$ of the clause. With the clauses sufficiently far to the right, we can write $c'-s_{18}$ as linear combination of $s_1-u_1$ and $s_2-u_2$ with positive coefficients. Therefore, we can scale $G_f^1$ and $G_f^2$ to move $s_{18}$ to $c'$. As above, 
  we get two $\Lambda$-cones of the same size.

\begin{figure}[htb]
  \centering
  \includegraphics{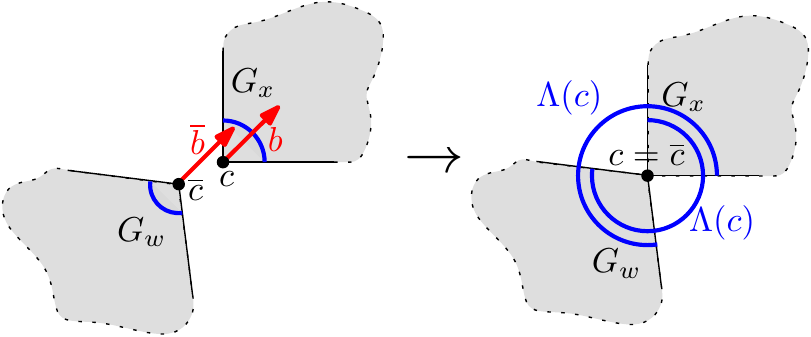}
  \caption{Connecting the gadget $G_w$ to $G_x$ such that we have two possible $\Lambda$-cones.}
  \label{fig:connecting_gadgets}
\end{figure}

We are now ready to  prove that \MSCtotEnergy and \MSClocEnergy are NP-hard.

\nphardness*
\begin{proof}
	Given an instance $I=(X,C)$ of \textsc{MNAE3SAT}, we construct $G_I$ with gadgets and auxiliary edges satisfying \cref{lem:clause-gadget,lem:variable-gadget,lem:wire-gadget}.
	Observe that $G_I$ is bipartite, because all gadgets are bipartite and the identified connector vertices can all be assigned the same color. Additionally, all vertices in $G_I$ have the same $\Lambda(v)$ by applying \cref{clm:lambdacone}.
	We show that $I$ is satisfiable if and only if $G_I$ has a \mincover. Because any optimal solution to \MSCtotEnergy and \MSClocEnergy is a \mincover, if a \mincover exists, this implies that both problems are NP-hard.

Suppose $I$ is satisfiable. Then there exists a
variable assignment~$\phi$ such that in no clause all values are equal. We construct a \mincover for $G_I$ by specifying a scan order. We start
with the variable gadgets. For each $x\in X$, scan all edges of $G_x$ with a
\mincover, as follows. If $\phi(x)=1$, the connector vertices rotate clockwise, otherwise counterclockwise.
By
\cref{lem:variable-gadget}.2, 
such a scan order exists.
Next, we scan all the wire gadgets $G_w$ such that both connector vertices of $G_w$ rotate in the same direction as the one attached to the vertex gadget did in the vertex gadget. By \cref{lem:wire-gadget}, such a scan order exists.
	Finally, scan the clause gadgets $G_c$. Each connector vertex already
rotated (counter-) clockwise, depending on $\phi(x)$ of the attached variable
$x$. Because $\phi$ is a valid (non-equal) assignment, no three connector vertices in
a clause gadget rotate in the same direction.
Thus, by \cref{lem:clause-gadget} a \mincover of this clause gadget exists.
	Therefore, all edges can be scanned by a \mincover.

It remains to show the converse.
Let $S$ be a \mincover of $G_I$.
We define $\varphi(v)=1$ if $v$ rotates clockwise over its $\Lambda$-cone, and $\varphi(v)=0$ otherwise.
We claim that if $x_c$ is a connector vertex in gadget
$G_x$, $\phi(x) = \varphi(x_c)$ yields a satisfying truth assignment to $I$.
	By \cref{lem:variable-gadget}, in any \mincover of $G_I$, all
$\varphi(x_c)$ are equal for all connector vertices in $G_x$, thus $\phi(x)$ is well defined. By \cref{lem:wire-gadget}, in any \mincover, both connector vertices in gadget $G_w$ rotate in the
same direction. Finally, by \cref{lem:clause-gadget}, no clause gadget
$G_c$ contains three connector vertices rotating in the same direction in a \mincover.
Therefore, $\phi(x)$ is a valid (non-equal) assignment of $I$, i.e., $I$ is satisfiable.
Because the reduction from $I$ to $G_I$ can be done in time polynomial in the number of clauses and variables, this concludes the proof.
\end{proof}

\nphardcor*
\begin{proof}
Given an instance $I$  of \textsc{MNAE3SAT}, consider the graph $G_I$ constructed in the proof of \cref{thm:NP-hard-2D}.
Recall that by \cref{clm:lambdacone}, we may assume that the angle of the $\Lambda$-\cone{s} coincide for all vertices. Define $\theta_{\max}:= \ang{360}-\Lambda(v)$ for some vertex $v$. Let $\theta_{\min}$ denote the minimal angle over all incident pairs of edges, for which both vertices have degree $>1$.

In a scan cover $S$ for $G_I$ that is not a \mincover, there exists a vertex $v$ with a total rotation angle exceeding $\Lambda(v)$.
Consequently, either $v$ scans its edges in a unidirectional rotation including the angle of size $\theta_{\max}$ (while possibly not scanning a smaller angle $\varepsilon$), or some angle between two edges at $v$ is covered at least twice.
We may assume that both vertices of those two edges have degree $>1$, because edges with a vertex of degree $1$ can be reordered freely to get an unidirectional rotation. Therefore, this angle has size at least $\theta_{\min}$.

Thus the objective value of $S$ is at least
$\min(\ang{360}-\varepsilon, \ang{360} -\theta_{\max} + \theta_{\min})$, while the cost of a \mincover is
$\ang{360}-\theta_{\max}$. Let $\alpha$ be the ratio between these two quantities. To prove the corollary, it is sufficient to prove $\alpha > 1.04$.

By checking all gadgets, we observe that $\theta_{\max} = \arctan(1/2)$
(the angle outside the $\Lambda$-\cone of $v_5$ in the wire fragment, see \cref{fig:swap-fragment}) and $\theta_{\min}\geq \arctan(1/4)$ (which is also realized at $v_5$). The angle $\varepsilon$ can be chosen smaller than any given constant by adding additional edges. This results in $\alpha \geq \frac{\ang{360} -\arctan(1/2) + \arctan(1/4)}{\ang{360} - \arctan(1/2)}\approx
1.042$.
\end{proof}

\end{document}